\documentclass[a4paper,UKenglish,cleveref,pdftex, thm-restate,numberwithinsect]{lipics-v2019}

\newif\ifreport
\reporttrue

\usepackage{amssymb}
\usepackage{mathtools}
\usepackage{mathpartir}
\usepackage[utf8]{inputenc}
\usepackage{tikz}
\usetikzlibrary{matrix, arrows.meta, positioning, calc}
\usepackage[usestackEOL]{stackengine}
\usepackage{proof-at-the-end}

\DeclareFontFamily{OT1}{pzc}{}
\DeclareFontShape{OT1}{pzc}{m}{it}{<-> s * [1.200] pzcmi7t}{}
\DeclareMathAlphabet{\mathpzc}{OT1}{pzc}{m}{it}
\renewcommand{\mathscr}{\mathpzc}

\newcommand{\catname}[1]{{\normalfont\textbf{#1}}}

\newcommand{\abs}[1]{\left\lvert#1\right\rvert}

\newcommand{\scat}[1]{\mathsf{suc}_{\gamma_{#1}}}
\newcommand{\scato}{\mathsf{suc}}
\newcommand{\diam}[1]{\mathsf{pre}_{\gamma_{#1}}}
\newcommand{\msub}[1]{\mathscr{M}\text{-}\mathscr{Sub}_{\catname{C}}}
\newcommand{\diamo}{\mathsf{pre}}

\newcommand{\sub}[1]{\mathscr{Sub}_{\catname{#1}}}
\newcommand{\dom}{\mathrm{dom}}
\newcommand{\cod}{\mathrm{cod}}

\newcommand{\lind}[1]{\mathcal{L}({#1})}
\newcommand{\reach}[1]{\mathsf{reach}_{#1}}
\newcommand{\integ}[1]{\mathsf{int}_{#1}}
\newcommand{\class}[1]{\catname{Cl}({#1})}
\newcommand{\propo}[0]{\mathsf{Prop}}

\newcommand{\formu}[1]{\mathbf{Form}_{#1}}

\newcommand{\escape}[3]{#2 \mathfrak{E}_{#1} #3}
\newcommand{\es}[3]{ \mathsf{isER}_{#1}(p,#2, #3)}

\newcommand{\er}[3]{\mathsf{EscR}_{#1}({#2},{#3})}
\newcommand{\coalg}[1]{\catname{CoAlg}(\mathscr{#1})}

\newcommand{\hoclass}[1]{\catname{HoCl}({#1})}
\newcommand{\term}[2]{\mathbf{Term}({#1}, {#2})}
\newcommand{\typ}[1]{\mathsf{Type}({#1})}
\newcommand{\sur}[3]{{#2}\mathfrak{S}_{#1}{#3}}
\newcommand{\fuz}[0]{\mathscr{FzSub}}

\newcommand\functorop[1][l]{\csname#1functor\endcsname}
\newcommand\lfunctorop[3]{%
	\setbox0=\hbox{$#2$}%
	\kern\wd0%
	\ensurestackMath{\Centerstack[c]{#1\\ \mathllap{#2\;\,}\mathclap{\DownArrow}\\#3}}%
}
\newcommand\rfunctorop[3]{%
	\setbox0=\hbox{$#2$}%
	\ensurestackMath{\Centerstack[c]{#1\\\mathclap{\UpArrow}\mathrlap{\,\;#2}\\#3}}%
	\kern\wd0%
}
\newcommand\functoropmapsto{\mathrel{\ensurestackMath{\Centerstack[c]{\longmapsto\\ \\\longmapsto}}}}
\setstackgap{L}{1.3\normalbaselineskip}
\newcommand\UpArrow{\rotatebox[origin=c]{90}{$\longrightarrow$\,}}
\newcommand\DownArrow{\rotatebox[origin=c]{-90}{$\longrightarrow$\,}}

\newcommand\functor[1][l]{\csname#1functor\endcsname}
\newcommand\lfunctor[3]{%
	\setbox0=\hbox{$#2$}%
	\kern\wd0%
	\ensurestackMath{\Centerstack[c]{#1\\ \mathllap{#2\;\,}\mathclap{\DownArrow}\\#3}}%
}
\newcommand\rfunctor[3]{%
	\setbox0=\hbox{$#2$}%
	\ensurestackMath{\Centerstack[c]{#1\\\mathclap{\DownArrow}\mathrlap{\,\;#2}\\#3}}%
	\kern\wd0%
}
\newcommand\functormapsto{\mathrel{\ensurestackMath{\Centerstack[c]{\longmapsto\\ \\\longmapsto}}}}
\setstackgap{L}{1.3\normalbaselineskip}

\newtheorem*{thm1*}{Theorem}

\newtheorem*{lem*}{Lemma}

\newtheorem*{prop*}{Proposition}

\title{Closure Hyperdoctrines\ifreport, with paths\fi}
\author{Davide Castelnovo}%
{Department of Mathematics, Computer Science and Physics, University of Udine, Italy}
{davide.castelnovo@uniud.it}{}{}
\author{Marino Miculan}
{Department of Mathematics, Computer Science and Physics, University of Udine, Italy}
{marino.miculan@uniud.it}{0000-0003-0755-3444}{Supported by the Italian MIUR  project PRIN 2017FTXR7S \emph{IT MATTERS (Methods and Tools for Trustworthy Smart Systems)}.}

\authorrunning{D.~Castelnovo, M.~Miculan}
\Copyright{Davide~Castelnovo and Marino Miculan}



\ifreport
\nolinenumbers %
\hideLIPIcs  %
\else
\relatedversion{An extended version is available at \url{https://arxiv.org/abs/2007.04213}.} %
\fi

\ifreport
\else
\EventEditors{John Q. Open and Joan R. Access}
\EventNoEds{2}
\EventLongTitle{42nd Conference on Very Important Topics (CVIT 2016)}
\EventShortTitle{CVIT 2016}
\EventAcronym{CVIT}
\EventYear{2016}
\EventDate{December 24--27, 2016}
\EventLocation{Little Whinging, United Kingdom}
\EventLogo{}
\SeriesVolume{42}
\ArticleNo{23}
\fi

\begin{document}
	
	\maketitle
	
	\begin{abstract}
		\emph{(Pre)closure spaces} are a generalization of topological spaces covering also the notion of neighbourhood in discrete structures, widely used to model and reason about spatial aspects of distributed systems.
				
	   \looseness=-1
		In this paper we introduce an abstract theoretical framework for the systematic investigation of the logical aspects of closure spaces.
		To this end, we introduce the notion of \emph{closure (hyper)doctrines}, i.e.~doctrines endowed with inflationary operators (and subject to suitable conditions).
		The generality and effectiveness of this concept is witnessed by many examples arising naturally from topological spaces, fuzzy sets, algebraic structures, coalgebras, and covering at once also known cases such as Kripke frames and probabilistic frames (i.e., Markov chains). 
		\ifreport
		Then, we show how spatial logical constructs concerning \emph{surroundedness} and  \emph{reachability} can be interpreted by endowing hyperdoctrines with a general notion of \emph{paths}. %
		\fi
		By leveraging general categorical constructions, we provide axiomatisations and sound and complete semantics for various fragments of logics for closure operators.
		
		Therefore, closure hyperdoctrines are useful both for refining and improving the theory of existing spatial logics, but especially for the definition of new spatial logics for new applications.
	\end{abstract}

	\section{Introduction}	
	Recently, much attention has been devoted in Computer Science to systems distributed in physical space; 
	a typical example is provided by the so called \emph{collective adaptive systems}, 
	such as drone swarms, sensor networks, autonomic vehicles, etc.
	This arises the problem of how to model and reason formally about spatial aspects of distributed systems. 
	To this end, several researchers have advocated the use of \emph{spatial logics}, i.e.~modal logics whose modalities are interpreted using topological concepts of neighbourhood and connectivity.\footnote{Not to be confused with spatial logics for reasoning on the structure of agents, such as the Ambient Logic \cite{CardelliG00} or the Brane Logic \cite{MiculanB06}.}
	
	In fact, the interpretation of modal logics in topological spaces goes back to Tarski;
	we refer to \cite{Aiello2007} for a comprehensive discussion of variants and computability and complexity aspects. 
	More recently, Ciancia \emph{et al.}~\cite{ciancia2014specifying,ciancia2016spatial} extended this approach to \emph{preclosure spaces}, also called \emph{\v{C}ech closure spaces}, which generalise topological spaces by not requiring idempotence of closure operator.
	This generalization unifies the notions of neighbourhood arising from topological spaces and from  \emph{quasi-discrete closure spaces}, like those induced by graphs and images.
	Building on this generalization, \cite{ciancia2014specifying} introduced \emph{Spatial Logic for Closure Spaces} (SLCS), a modal logic for the specification and verification on spatial concepts over preclosure spaces.
	This logic features a \emph{closure} modality and a spatial \emph{until} modality: intuitively $\phi \mathcal{U} \psi$ holds in an area where $\phi$ holds and it is not possible to ``escape'' from it unless passing through an area where $\psi$ holds.
	\ifreport
	There is also a \emph{surrounded} constructor, to represent a notion of (un)reachability.  
	\fi
	SLCS has been proved to be quite effective and expressive, as it has been applied to reachability problems, vehicular movement, digital image analysis (e.g., street maps, radiological images \cite{BelmonteCLM19}), etc. 
	The model checking problem for this logic over finite quasi-discrete structures is decidable in linear time \cite{ciancia2014specifying}.
	
	Despite these results, an axiomatisation for SLCS is still missing.
	Moreover, it is not obvious how to extend this logic to other spaces with closure operators, such as probabilistic automata  (e.g.~Markov chains).
	\ifreport
	Also, it is not immediate how current definitions of reachability can be generalized to other cases, e.g., within a given number of steps, or non-deterministic, or probabilistic, etc.
	\fi

	In fact, the main point is that we miss an abstract theoretical framework for investigating the logical aspects of (pre)closure spaces. Such a framework would be the basis for analysing spatial logics like SLCS, but also for developing further extensions and applications thereof.

	In this paper, we aim to build such a framework. To this end, we introduce the new notion of \emph{closure (hyper)doctrine} as the theoretical basis for studying the logical aspects of closure spaces.
	Doctrines were introduced by Lawvere \cite{lawvere1969adjointness} as a general way for endowing (the objects of) a category with logical notions from a suitable 2-category $\catname{E}$, which can be the category of Heyting algebras in the case of intuitionistic logic, of Boolean algebras in the case of classical logic, etc..
	Along this line, in order to capture the logical aspects of closure spaces we introduce the notion of \emph{closure operators} on doctrines, that is, families of inflationary morphisms over objects of  $\catname{E}$ (subject to suitable conditions); a closure (hyper)doctrine is a (hyper)doctrine endowed with a closure operator.
	These structures arise from many common situations: we provide many examples ranging from topology to algebraic structures, from coalgebras to fuzzy sets. These examples cover the usual cases from literature (e.g., graphs, quasi-discrete spaces, (pre)topological spaces) but include also new settings, such as categories of coalgebras and probabilistic frames (i.e., Markov chains).
	Then, leveraging general machinery from categorical logic, we introduce a first order logic for closure spaces for which we provide an axiomatisation and a sound and complete categorical semantics.  The propositional fragment corresponds to the SLCS from \cite{ciancia2014specifying}.
	
	\ifreport
	Within this framework, we can accommodate also the notion of \emph{surroundedness} of properties, in order to model spatial operators like SLCS's $\mathcal{S}$ \cite{ciancia2016spatial}.
	Actually, surroundedness is not a structural property of the logical domain (differently from closure operators); rather, it depends on the kind of paths we choose to explore the space. To this end, we introduce the notion of \emph{closure doctrine with paths}. Again, the foundational approach we follow allows for many kinds of paths, and hence many notions of surroundedness.
	\fi
	
	Overall, the importance of this work is twofold: on one hand, closure hyperdoctrines \ifreport (with paths) \fi are useful for analysing and improving the theory of existing spatial logics; in particular, the proposed axiomatisation can enable both new proof methodologies and minimisation techniques. On the other, closure hyperdoctrines are useful for the definition of new logics in various situations where we have to deal with closure operators, connectivity, surroundedness, reachability, etc.

	\smallskip

	\noindent\textit{Synopsis}
	The paper is organized as follows.
	In \cref{sec:closurehd} we recall (hyper)doctrines and introduce the key notion of closure doctrine.
	Many examples of closure doctrines are provided in \cref{sec:examples}.
	In \cref{sec:slcs} we introduce \emph{logics for closure operators}, together with a sound and complete semantics in closure hyperdoctrines.
	\ifreport
	Then, in order to cover the notion of surroundedness, we introduce the notion of \emph{closure doctrine with paths} (\cref{sec:paths}) and the corresponding logics with the ``surrounded'' operator (\cref{sec:slcswp}).
	\fi
	Conclusions and directions for future work are in \cref{sec:concl}.
	Longer proofs are in \cref{sec:proofs}.

	\section{Closure (hyper)doctrines}
	\label{sec:closurehd}
	
	\subsection{Kinds of doctrines}
	\label{sec:background}
	In this section we recall the notion of elementary hyperdoctrine, due to Lawvere \cite{lawvere1969adjointness,lawvere1970equality}.
	The development of semantics of logics in this context or in the equivalent fibrational context is well established; we refer the reader to, e.g., \cite{jacobs1999categorical,makkai2006first,pitts1995categorical}.

	\begin{definition}[(Existential) Doctrine, Hyperdoctrine \cite{kock1977doctrines,maietti2013quotient,pasquali2015co}]%
		A \emph{primary doctrine} or simply a \emph{doctrine} on a category $\catname{C}$ is a functor $\mathscr{P}:\catname{C}^{op}\rightarrow \catname{InfSL}$ where $\catname{InfSL}$ is the category of finite meet semilattices.
		
		A primary doctrine is \emph{existential} if:
		\begin{itemize}
			\item  $\catname{C}$ has finite products;
			\item the image $\mathscr{P}_{\pi_C}$ of any projection $\pi_C: C\times D\rightarrow C$ admits a left adjoint $\exists_{\pi_C}$;
			\item 	\parbox[t]{75mm}{for each pullback like aside, the \emph{Beck-Chevalley condition} $\exists_{\pi_{C'}}\circ \mathscr{P}_{1_D\times f}=\mathscr{P}_f\circ \exists _{\pi_C}$ holds;}
			\begin{tikzpicture}[baseline=(current bounding box.center)]
			\node(A) at(0,2) {$D\times C'$};
			\node(B) at(0,0) {$D\times C$};
			\node(C) at(2,2) {$C'$};
			\node(D) at(2,0) {$C$};
			\draw[->] (A)--(B) node[pos=0.5, left]{$1_D\times f$};
			\draw[->] (A)--(C) node[pos=0.5, above]{$\pi_{C'}$};
			\draw[->] (C)--(D) node[pos=0.5, right]{$f$};
			\draw[->] (B)--(D) node[pos=0.5, below]{$\pi_C$};
			\end{tikzpicture}
			
			\item  for any $\alpha \in \mathscr{P}(C)$ and $\beta \in \mathscr{P}(D\times C)$ the \emph{Frobenius reciprocity}
			$\exists_{\pi_C}(\mathscr{P}_{\pi_C}(\alpha)\wedge \beta)=\alpha \wedge \exists_{\pi_C}(\beta)$
			holds.
		\end{itemize}
		A \emph{hyperdoctrine} is an existential doctrine $\mathscr{P}$ such that:
		\begin{itemize}
			\item $\mathscr{P}$ factors through the category $\catname{HA}$ of Heyting algebras and Heyting algebras morphisms;
			\item for all projections $\pi_C:D\times C\rightarrow C$, $\mathscr{P}_{\pi_C}$ has a right adjoint $\forall_{\pi_C}:\mathscr{P}(D\times C)\to \mathscr{P}(C)$ which must satisfy the Beck-Chevalley condition:
			$\forall_{\pi_{C'}}\circ \mathscr{P}_{1_D\times f}=\mathscr{P}_f\circ \forall _{\pi_C}$
			for any $f:C'\rightarrow C$.
		\end{itemize}
		A primary doctrine, an existential doctrine or a hyperdoctrine, is \emph{elementary} if
		\begin{itemize}
			\item $\catname{C}$ has finite products;
			\item for each object $C$ there exists a \emph{fibered equality}  $\delta_C \in \mathscr{P}(C\times C)$  such that \begin{equation*}
			\mathscr{P}_{(\pi_1, \pi_2)}(-)\wedge \mathscr{P}_{(\pi_2, \pi_3)}(\delta_C) \dashv \mathscr{P}_{1_D\times \Delta_C}
			\end{equation*} 	
			
			where $\pi_1, \pi_2$ and $\pi_3$ are projections $D\times C\times C\rightarrow D\times C$.
			This left adjoint will be denoted by $\exists_{1_D\times \Delta_C}$
		\end{itemize} 
	\end{definition}
	
	\begin{remark}
		Usually $\catname{C}$ is required to having finite products even in the case of a primary doctrine (cfr. \cite{pasquali2015co}), we will not ask it in order to get the coalgebraic examples in \cref{sec:examples}.
	\end{remark}
	
	\begin{remark}
		Since $\catname{C}$ has a terminal object it follows that 
		$\mathscr{P}_{\pi_1}(-)\wedge \delta_C \dashv \mathscr{P}_{\Delta_C}$.
		This left adjoint will be denoted by $\exists_{\Delta_C}$.
	\end{remark}
	
	\begin{remark}
		In this paper, we work with hyperdoctrines over $\catname{HA}$, the category of Heyting algebras and their morphisms; hence the resulting logic is inherently intuitionistic.
		Clearly, all the development still holds if we restrict ourselves to the subcategory of Boolean algebras $\catname{BA}$, yielding a classical version of the logic.	
	\end{remark}
	\begin{example}
		\looseness=-1
		Let $\catname{C}$ be a category with finite limits and $(\mathscr{E}, \mathscr{M})$ a stable and proper factorization system on it (see \cite{kelly1991note}). Fix an object $C\in \catname{C}$  we define a relation on arrows in $\mathscr{M}$ with codomain $C$ putting $m\leq n$ if and only if there exists $t$ such that $n\circ t =m$. If we ignore size issues this gives us a preorder, from which we get a partial order $\msub{C}(C)$ by quotienting by the relation $m\simeq n$ if and only if $m\leq n$ and $n\leq m$. The top element is $[1_C]$, while meets are given by pullbacks, and we can pullback any $m$ along any arrow $f:D\rightarrow C$ getting an arrow $f^*m$ in $\mathscr{M}$ with codomain $D$. Summarizing we have a functor $	\catname{C}^{op}\rightarrow \catname{InfSL}$ sending $C$ to $\msub{C}(C)$.
		This is actually an elementary existential doctrine in which $\delta_C$ is the class of the diagonal $C\rightarrow C\times C$ (which can be shown to be an element of $\mathscr{M}$) and $\exists_{\pi_C}([m])$ is the $\mathscr{M}$-component of $\pi_C\circ m$, in the sense that it is the class of $n\in \mathscr{M}$ such that  $n\circ e=\pi_C\circ m$ for some $e\in \mathscr{E}$ (see \cite{hughes2003factorization} for the correspondence between factorization systems and elementary existential doctrines).
		In general this functor  is very far from having Heyting algebras as values but this is the case when $\catname{C}$ is a topos and $\mathscr{M}$ the class of all monomorphisms; in this case we get an elementary hyperdoctrine \cite{maclane2012sheaves}.
	\end{example}

If we want $\mathscr{M}$ to be the class of all monos we have the following theorem.
\begin{theorem}\label{reg}
If $\catname{C}$ has finite limits then $\sub{C}$ is an elementary existetial doctrine if and only if $\catname{C}$ is regular.
\end{theorem}
\begin{proof}
	Cfr. \cite{jacobs1999categorical}, theorem $4.4.4$. \qedhere 
\end{proof}
		\begin{theoremEnd}{prop}\label{trans}
		Let $\mathscr{P}:\catname{C}^{op}\rightarrow \catname{InfSL}$ be an existential doctrine, $\catname{D}$ a category with finite products and $\mathscr{F}:\catname{D}\rightarrow \catname{C}$ a product preserving functor. 
		Then, $\mathscr{P}\circ \mathscr{F}^{op}$ is a existential doctrine.
		If $\mathscr{P}$ is elementary (resp., a hyperdoctrine) then $\mathscr{P}\circ \mathscr{F}^{op}$ is elementary (resp., a hyperdoctrine).
	\end{theoremEnd}
\begin{proof}
	\begin{proofEnd}	We have to show that $\mathscr{P}_{\mathscr{F}(\pi_D)}$ has a left adjoint for any projection $\pi_D:E\times D\rightarrow D$ but this follows at once since the horizontal arrow in the diagram
		\begin{center}
			\begin{tikzpicture}
			\node(F)at(0,0){$\mathscr{P}(\mathscr{F}(E\times D))$};
			\node(I)at(6,0){$\mathscr{P}(\mathscr{F}(E)\times \mathscr{F}(D))$};
			\node(G)at(3,2){$\mathscr{P}(\mathscr{F}(D))$};
			\draw[<-](F)--(G)node[pos=0.5, above, xshift=-0.4cm]{$\mathscr{P}_{\mathscr{F}(\pi_D)}$};
			\draw[<-](I)--(G)node[pos=0.5, above, xshift=0.2cm, yshift=-0.05cm]{$\mathscr{P}_{\pi_{\mathscr{F}(D)}}$};
			\draw[<-](I)--(F)node[pos=0.5,below]{$\mathscr{P}_{(\mathscr{F}(\pi_E),\mathscr{F}(\pi_D))}$};
			\end{tikzpicture}
		\end{center}
		is an isomorphism, hence
		\begin{equation*}
		\exists_{\pi_D}=\exists_{\pi_{\mathscr{F}(D)}}\circ \mathscr{P}_{(\mathscr{F}(\pi_E),\mathscr{F}(\pi_D))}
		\end{equation*} The same argument shows the existence of right adjoint to $\mathscr{P}_{\mathscr{F}(\pi_D)}$ anytime they exist for $\mathscr{P}_{\pi_{\mathscr{F}(D)}}$:
		\begin{equation*}
		\forall_{\pi_D}=\forall_{\pi_{\mathscr{F}(D)}}\circ \mathscr{P}_{(\mathscr{F}(\pi_E),\mathscr{F}(\pi_D))}
		\end{equation*}  Let $f:D'\rightarrow D$ be an arrow in $\catname{D}$, the two Beck-Chevalley conditions follow from the commutativity of
		\begin{center}
			\begin{tikzpicture}
			\node(F)at(0,2){$\mathscr{P}(\mathscr{F}(E\times D))$};
			\node(I)at(0,0){$\mathscr{P}(\mathscr{F}(E)\times \mathscr{F}(D))$};
			\node(G)at(6,0){$\mathscr{P}(\mathscr{F}(D))$};
			
			\node(F')at(0,-4){$\mathscr{P}(\mathscr{F}(E\times D'))$};
			\node(I')at(0,-2){$\mathscr{P}(\mathscr{F}(E)\times \mathscr{F}(D'))$};
			\node(G')at(6,-2){$\mathscr{P}(\mathscr{F}(D'))$};
			\draw[<-](F)..controls(4,1.5)..(G)node[pos=0.5, above, yshift=0cm]{$\mathscr{P}_{\mathscr{F}(\pi_D)}$};
			\draw[<-](I)--(G)node[pos=0.5, above, xshift=0cm, yshift=0cm]{$\mathscr{P}_{\pi_{\mathscr{F}(D)}}$};
			\draw[<-](I)--(F)node[pos=0.5, right]{$\mathscr{P}_{(\mathscr{F}(\pi_E),\mathscr{F}(\pi_D))}$};

			\draw[<-](F')..controls(4,-3.5)..(G')node[pos=0.5, below, yshift=-0.1cm]{$\mathscr{P}_{\mathscr{F}(\pi_{D'})}$};
			\draw[<-](I')--(G')node[pos=0.5, below]{$\mathscr{P}_{\pi_{\mathscr{F}(D')}}$};
			\draw[<-](I')--(F')node[pos=0.5,right]{$\mathscr{P}_{(\mathscr{F}(\pi_E),\mathscr{F}(\pi_D))}$};
			
			\draw[<-](G)--(G')node[pos=0.5,right]{$\mathscr{P}_{\mathscr{F}(f)}$};
			\draw[<-](I)--(I')node[pos=0.5,left]{$\mathscr{P}_{1_{\mathscr{F}(E)}\times \mathscr{F}(f)}$};
			\draw[->](F')..controls(-3,-4) and (-3, 2)..(F)node[pos=0.5,left]{$\mathscr{P}_{\mathscr{F}(1_E\times f)}$};
			\end{tikzpicture}
		\end{center}
		and the fact that both the upper and the lower vertical arrow are isomorphisms since $\mathscr{F}$ preserves products.
		For Frobenius reciprocity:
		\begin{align*}
		\exists_{\mathscr{F}(\pi_D)}(\mathscr{P}_{\mathscr{F}(\pi_D)}(\alpha)\wedge \beta)&= \exists_{\pi_{\mathscr{F}(D)}}(\mathscr{P}_{(\mathscr{F}(\pi_E),\mathscr{F}(\pi_D))}(\mathscr{P}_{\mathscr{F}(\pi_D)}(\alpha)\wedge \beta))\\&=
		\exists_{\pi_{\mathscr{F}(D)}}(\mathscr{P}_{\pi_{\mathscr{F}(D)}}(\alpha)\wedge \mathscr{P}_{(\mathscr{F}(\pi_E),\mathscr{F}(\pi_D))}(\beta))\\&=\alpha \wedge \exists_{\pi_{\mathscr{F}(D)}}(\mathscr{P}_{\mathscr{F}(\pi_D)}(\alpha)\wedge \beta)\\&= \alpha \wedge 	\exists_{\mathscr{F}(\pi_D)}(\beta)
		\end{align*}
		So we're left with the fibered equalities, by the commutativity of
		\begin{center}
			\begin{tikzpicture}
			\node(A) at(0,2) {$\mathscr{P}(\mathscr{F}(E\times D \times D))$};
			\node(C) at(8,2) {$\mathscr{P}(\mathscr{F}(E)\times \mathscr{F}(D) \times \mathscr{F}(D))$};
			\node(B) at(0,0) {$\mathscr{P}(\mathscr{F}(E\times D))$};
			\node(D) at(8,0) {$\mathscr{P}(\mathscr{F}(E)\times \mathscr{F}(D))$};
			\draw[->] (A)--(B) node[pos=0.5, left]{$\mathscr{P}_{\mathscr{F}(1_E\times \Delta_D)}$};
			\draw[<-] (A)--(C) node[pos=0.5, above]{$\mathscr{P}_{(\mathscr{F}(\pi_1),\mathscr{F}(\pi_2),\mathscr{F}(\pi_3))}$};
			\draw[->] (C)--(D) node[pos=0.5, right]{$\mathscr{P}_{1_{\mathscr{F}(E)}\times \Delta_{\mathscr{F}(D)}}$};
			\draw[<-] (B)--(D) node[pos=0.5, below]{$\mathscr{P}_{(\mathscr{F}(\pi_E), \mathscr{F}(\pi_D))}$};
			
			\node(B') at(0,-2) {$\mathscr{P}(\mathscr{F}(D\times D))$};
			\node(D') at(8,-2) {$\mathscr{P}(\mathscr{F}(D)\times \mathscr{F}(D))$};
			\node(A) at(0,-4) {$\mathscr{P}(\mathscr{F}(E\times D \times D))$};
			\node(C) at(8,-4) {$\mathscr{P}(\mathscr{F}(E)\times \mathscr{F}(D) \times \mathscr{F}(D))$};
			\node(B) at(0,-6) {$\mathscr{P}(\mathscr{F}(E\times D))$};
			\node(D) at(8,-6) {$\mathscr{P}(\mathscr{F}(E)\times \mathscr{F}(D))$};
			\draw[<-] (A)--(B) node[pos=0.5, left]{$\mathscr{P}_{\mathscr{F}(\pi_1, \pi_2)}$};
			\draw[<-] (A)--(B') node[pos=0.5, left]{$\mathscr{P}_{\mathscr{F}(\pi_2, \pi_3)}$};
			\draw[<-] (A)--(C) node[pos=0.5, above]{$\mathscr{P}_{(\mathscr{F}(\pi_1),\mathscr{F}(\pi_2),\mathscr{F}(\pi_3))}$};
			\draw[<-] (C)--(D) node[pos=0.5, right]{$\mathscr{P}_{(\mathscr{F}(\pi_1),\mathscr{F}(\pi_2))}$};
			\draw[<-] (C)--(D') node[pos=0.5, right]{$\mathscr{P}_{(\mathscr{F}(\pi_2),\mathscr{F}(\pi_3))}$};
			\draw[<-] (B)--(D) node[pos=0.5, below]{$\mathscr{P}_{(\mathscr{F}(\pi_E), \mathscr{F}(\pi_D))}$};
			\draw[<-] (B')--(D') node[pos=0.5, above]{$\mathscr{P}_{(\mathscr{F}(p_1), \mathscr{F}(p_2))}$};

			\end{tikzpicture}
		\end{center}
		and from the fact that the horizontal arrows are isomorphisms it is enough to take
		\begin{equation*}
		\delta_D:=\mathscr{P}_{(\mathscr{F}(p_1), \mathscr{F}(p_2))}(\delta_{\mathscr{F}(D)})
		\end{equation*}\qedhere 
	\end{proofEnd}\qedhere
\end{proof}

		\begin{theoremEnd}{prop}\label{lad}
		Let $\mathscr{P}:\catname{C}^{op}\rightarrow \catname{HA}$ be an elementary existential doctrine. For every arrow $f:C\rightarrow D$, the functor $\mathscr{P}_f$ has a left adjoint $\exists_f$ that satisfies the \emph{Frobenius reciprocity}:
		$\exists_{f}(\mathscr{P}_{f}(\beta)\wedge \alpha )=\beta\wedge \exists_{f}(\alpha)$.
		If $\mathscr{P}$ is a hyperdoctrine then $\mathscr{P}_f$ has a right adjoint $\forall_f$ too.
\end{theoremEnd}
\begin{proof}
		\begin{proofEnd}(Cfr. \cite{jacobs1999categorical,lawvere1970equality} and lemma $1.5.8$ of \cite{johnstone2002sketches}, vol. $1$ for the hyperdoctrine case).
		It is enough to define
			\begin{align*}
				\exists_f(\alpha)&:= \exists_{\pi_D}(\mathscr{P}_{f\times 1_D}(\delta_D)\wedge \mathscr{P}_{\pi_C}(\alpha) ) \\\forall_f(\alpha)&:= \forall_{\pi_D}(\mathscr{P}_{f\times 1_D}(\delta_D)\rightarrow \mathscr{P}_{\pi_C}(\alpha) ) 
			\end{align*}
			For adjointness:	
			\begin{itemize}
				\item $\exists_f \dashv \mathscr{P}_f$. 
				\begin{equation*}
				\begin{split}
				\text{If} \ \exists_f(\alpha)&\leq \beta\\
				\alpha &= \alpha \wedge \top_C \\
				&= \alpha \wedge \exists_{\pi_2}(\delta_C)\\
				&\leq \alpha \wedge \exists_{\pi_2}(\mathscr{P}_{f\times f}(\delta_D))\\
				&=	\exists_{\pi_2}(\mathscr{P}_{f\times f}(\delta_D)\wedge \mathscr{P}_{p_2}(\alpha))\\
				&=\exists_{\pi_2}(\mathscr{P}_{1_C\times f}(\mathscr{P}_{f\times 1_D}(\delta_D)\wedge \mathscr{P}_{\pi_C}(\alpha ))\\
				&=\mathscr{P}_f(\exists_{\pi_D}(\mathscr{P}_{f\times 1_D}(\delta_D)\wedge \mathscr{P}_{\pi_C}(\alpha) ))\\
				&=\mathscr{P}_f(\exists_f(\alpha))\\
				&\leq \mathscr{P}_f(\beta)\\
				&
				\end{split}
				\begin{split}
				\text{If} \ \alpha&\leq \mathscr{P}_f(\beta)\\
				\exists_f(\alpha)&\leq \exists_f(\mathscr{P}_f(\beta))\\&=
				\exists_{\pi_D}(\mathscr{P}_{f\times 1_D}(\delta_D)\wedge \mathscr{P}_{\pi_C}(\mathscr{P}_f(\beta) )\\&=\exists_{\pi_D}(\mathscr{P}_{f\times 1_D}(\delta_D)\wedge \mathscr{P}_{1_D\times f}(\mathscr{P}_{q_2}(\beta) ))\\&=
				\exists_{\pi_D}(\mathscr{P}_{ 1_D\times f}(\delta_D)\wedge \mathscr{P}_{1_D\times f}(\mathscr{P}_{q_2}(\beta) ))
				\\&=
				\exists_{\pi_D}(\mathscr{P}_{ 1_D\times f}(\delta_D\wedge \mathscr{P}_{q_2}(\beta)))
				\\&=\exists_{\pi_D}(\mathscr{P}_{ 1_D\times f}(\exists_{\Delta_D}(\beta)))\\&\leq\exists_{\pi_D}(\mathscr{P}_{ 1_D\times f}(\mathscr{P}_{q_1}(\beta)))\\&=\exists_{\pi_D}(\mathscr{P}_{\pi_D}(\beta))\\&\leq \beta 
				\end{split}		
				\end{equation*}
				Where $p_2$ is the second projection $C\times C\rightarrow C$ and  $q_1$ and $q_2$ those $D\times D\rightarrow D$.
				\item $\mathscr{P}_f\dashv \forall_f$. We already know that:
				\begin{equation*}
				\exists_{\pi_C}\dashv \mathscr{P}_{\pi_C} \quad \mathscr{P}_{f\times 1_D}(\delta_D)\wedge (-) \dashv \mathscr{P}_{f\times 1_D}(\delta_D)\rightarrow (-) \quad \mathscr{P}_{\pi_D}\dashv \forall_{\pi_D}
				\end{equation*}
				so it is enough to show, for all $\beta \in \mathscr{P}(D)$
				\begin{equation*}
				\mathscr{P}_f(\beta) = \exists_{\pi_C}(\mathscr{P}_{f\times 1_D}(\delta_D)\wedge\mathscr{P}_{\pi_D}(\beta))
				\end{equation*}
				but this is easily done:
				\begin{align*}
				\exists_{\pi_C}(\mathscr{P}_{f\times 1_D}(\delta_D)\wedge\mathscr{P}_{\pi_D}(\beta))&= 	\exists_{\pi_C}(\mathscr{P}_{1_D\times f}(\delta_D)\wedge \mathscr{P}_{1_D\times f}(\mathscr{P}_{\pi_1}(\beta)))\\&=
				\exists_{\pi_C}(\mathscr{P}_{1_D\times f}(\delta_D\wedge \mathscr{P}_{\pi_1}(\beta)))\\&=
				\mathscr{P}_f(\exists_{\pi_2}(\exists_{\Delta_D}(\beta)))\\&=\mathscr{P}_f(\beta)
				\end{align*}
				Where $\pi_2$ is the second projection $D\times D\rightarrow D$.	
			\end{itemize}
			For Frobenius reciprocity:  the inequality $	\exists_{f}(\mathscr{P}_{f}(\beta)\wedge \alpha )\leq \beta\wedge \exists_{f}(\alpha)$ follows from adjointness, let's show the other. Let $\pi_1$ and $\pi_2$ be the projections from $D\times D$ and compute:
			\begin{align*}
			\exists_{f}(\mathscr{P}_{f}(\beta)\wedge \alpha )&= \exists_{\pi_D}(\mathscr{P}_{f \times 1_D}(\delta_D)\wedge \mathscr{P}_{\pi_C}(\mathscr{P}_f(\beta)\wedge \alpha))\\&=\exists_{\pi_D}(\mathscr{P}_{f \times 1_D}(\delta_D)\wedge \mathscr{P}_{\pi_C}(\mathscr{P}_f(\beta))\wedge \mathscr{P}_{\pi_C}(\alpha))\\&=
			\exists_{\pi_D}(\mathscr{P}_{f \times 1_D}(\delta_D)\wedge \mathscr{P}_{f\times 1_D}(\mathscr{P}_{\pi_1}(\beta))\wedge \mathscr{P}_{\pi_C}(\alpha))\\&=
			\exists_{\pi_D}(\mathscr{P}_{f\times 1_D}(\delta_D\wedge \mathscr{P}_{\pi_1}(\beta)) \wedge \mathscr{P}_{\pi_C}(\alpha))\\&\leq 
			\exists_{\pi_D}(\mathscr{P}_{f\times 1_D}(\delta_D\wedge \mathscr{P}_{\pi_2}(\beta)) \wedge \mathscr{P}_{\pi_C}(\alpha))\\&=
			\exists_{\pi_D}(\mathscr{P}_{f\times 1_D}(\delta_D)\wedge \mathscr{P}_{f\times 1_D}(\mathscr{P}_{\pi_2}(\beta)) \wedge \mathscr{P}_{\pi_C}(\alpha))\\&=
			\exists_{\pi_D}(\mathscr{P}_{f\times 1_D}(\delta_D)\wedge \mathscr{P}_{\pi_D}(\beta) \wedge \mathscr{P}_{\pi_C}(\alpha))\\&=\exists_{\pi_D}(\mathscr{P}_{f\times 1_D}(\delta_D) \wedge \mathscr{P}_{\pi_C}(\alpha))\wedge \beta \\&=
			\exists_f(\alpha)\wedge \beta 
			\end{align*}
			where we have used the inequality:
			\begin{equation*}
			\delta_D\wedge \mathscr{P}_{\pi_1}(\beta) \leq \mathscr{P}_{\pi_2}(\beta)
			\end{equation*}
			that follows at once by the definition of  $\exists_{\Delta_D}$.	
			\qedhere 
		\end{proofEnd}
	\qedhere 
\end{proof}
\begin{remark}
	In general these adjoints do not satisfy any form of Beck-Chevalley condition \cite{vcubric1997semantics,jacobs1999categorical,maietti2017triposes,seely1983hyperdoctrines}.
\end{remark}
	\begin{definition}
		Let  $\mathscr{P}:\catname{C}^{op}\rightarrow \catname{InfSL}$, $\mathscr{S}:\catname{D}^{op}\rightarrow \catname{InfSL}$ be primary doctrines.
		
		A morphism $\mathscr{P}\rightarrow \mathscr{S}$ is a pair $(\mathscr{F}, \eta )$ where $\mathscr{F}:\catname{C}\to\catname{D}$ is a functor and $\eta:\mathscr{P}\rightarrow \mathscr{S}\circ \mathscr{F}^{op}$ is a natural transformation.
		
		$(\mathscr{F}, \eta)$ is a \emph{morphism of elementary doctrines}, or \emph{elementary}, if $\mathscr{F}$ preserves finite products and for any object $C$ of $\catname{C}$, it is
		$\eta_{C\times C}(\delta_C)=\mathscr{S}_{(\mathscr{F}(\pi_1), \mathscr{F}(\pi_2))}(\delta_{\mathscr{F}(C)})$.

		$(\mathscr{F}, \eta)$ is a \emph{morphism of existential doctrine} if $\mathscr{F}$ preserves finite products and for any pair of objects $C,D$ of $\catname{C}$  the diagram $(a)$ below commutes.
	    \vspace{-1ex}
		\[
		\begin{tikzpicture}[baseline=(current bounding box.center),scale=0.75]
		\node(E)at(-4,0){$\mathscr{P}(D\times C)$};
		\node(F)at(-4,-1.5){$\mathscr{S}(\mathscr{F}(D\times C))$};
		\node(I)at(-4,-3){$\mathscr{S}(\mathscr{F}(D)\times \mathscr{F}(C))$};
		\node(G)at(0,-3){$\mathscr{S}(\mathscr{F}(D))$};
		\node(H)at(0,0){$\mathscr{P}(C)$};
		\draw[->](E)--(H)node[pos=0.5, above]{$\exists_{\pi_C}$};
		\draw[->](E)--(F)node[pos=0.5, left]{$\eta_{D\times C}$};
		\draw[->](I)--(G)node[pos=0.5, above]{$\exists_{\pi_{\mathscr{F}(C)}}$};
		\draw[<-](I)--(F)node[pos=0.5, left]{$\mathscr{S}_{(\mathscr{F}(\pi_D),\mathscr{F}(\pi_C))^{}}$};
		\draw[->](H)--(G)node[pos=0.5, right]{$\eta_C$};
	    \node(L)at(-2,-4){(a)};
		\end{tikzpicture}%
	\hspace{1cm}
		\begin{tikzpicture}[baseline=(current bounding box.center),scale=0.75]
		\node(E)at(-4,-5){$\mathscr{P}(D\times C)$};
		\node(F)at(-4,-6.5){$\mathscr{S}(\mathscr{F}(D\times C))$};
		\node(I)at(-4,-8){$\mathscr{S}(\mathscr{F}(D)\times \mathscr{F}(C))$};
		\node(G)at(0,-8){$\mathscr{S}(\mathscr{F}(D))$};
		\node(H)at(0,-5){$\mathscr{P}(C)$};
		\draw[->](E)--(H)node[pos=0.5, above]{$\forall_{\pi_C}$};
		\draw[->](E)--(F)node[pos=0.5, left]{$\eta_{D\times C}$};
		\draw[->](I)--(G)node[pos=0.5, above]{$\forall_{\pi_{\mathscr{F}(C)}}$};
		\draw[<-](I)--(F)node[pos=0.5, left]{$\mathscr{S}_{(\mathscr{F}(\pi_D),\mathscr{F}(\pi_C))^{}}$};
		\draw[->](H)--(G)node[pos=0.5, right]{$\eta_C$};
		\node(L)at(-2,-9){(b)};
		\end{tikzpicture}
	\]

		$(\mathscr{F}, \eta)$ is a \emph{morphism of hyperdoctrines} if it is a morphism of existential doctrine, 
			the diagram $(b)$ above commutes too and each component of $\eta$ preserves finite suprema and implication.%
				
		If $(F,\eta)$ is also elementary then we call it a \emph{morphism of elementary existential doctrines} or of \emph{elementary hyperdoctrines}.
		
		Let $(\mathscr{F},\eta ), (\mathscr{G}, \epsilon): \mathscr{P}\rightarrow \mathscr{S}$ be two morphisms; a $2$-arrow $(\mathscr{F},\eta )\rightarrow (\mathscr{G}, \epsilon)$ is a natural transformations $\theta:\mathscr{F} \rightarrow \mathscr{G}$ such that
		$\eta_C(\alpha)\leq \mathscr{S}_{\theta_C}(\epsilon_C(\alpha))$.
		
		This defines the $2$-categories $\catname{PD}$, $\catname{ED}$, $\catname{HD}$ of primary doctrines, existential doctrines and hyperdoctrines, and the subcategories $\catname{EPD}$, $\catname{EED}$, $\catname{EHD}$ of their elementary variants.
	\end{definition}

	\subsection{Closure operators on doctrines}
	In this section we introduce the key notion of closure operators on doctrines.
	\begin{definition}
		Let $\mathscr{P}$ be a doctrine.
		A \emph{closure operator} on $\mathscr{P}$ is a (possibly large) family $\mathfrak{c} = \{\mathfrak{c}_C\}_{C\in \catname{Ob(C)}}$ of functions $\mathfrak{c}_C:\mathscr{P}(C)\rightarrow\mathscr{P}(C)$ such that:
		\begin{itemize}
			\item for any object $C$, $\mathfrak{c}_C$ is monotone and \emph{inflationary}, i.e., $1_{\mathscr{P}(C)} \leq \mathfrak{c}_C$
			\item any arrow $f:C\rightarrow D$ is \emph{continuous}, i.e.
			\begin{equation*}
				\mathfrak{c}_C\circ \mathscr{P}_f \leq \mathscr{P}_f \circ \mathfrak{c}_D
			\end{equation*}
		\end{itemize}
		A closure operator $\mathfrak{c}$ is said to be
		\begin{itemize}
			\item \emph{grounded} if  $\mathfrak{c}_C(\bot)=\bot$
			for all objects $C$ such that $\mathscr{P}(C)$ has a minimum;
			\item \emph{additive} if
			\begin{equation*}
			\mathfrak{c}_C(\alpha \vee \beta)=\mathfrak{c}_C(\alpha) \vee \mathfrak{c}_C(\beta)
			\end{equation*} 
			for all objects $C$ such that $\mathscr{P}(C)$ has binary suprema;
			\item \emph{finitely additive} if it is grounded and additive;
			\item \emph{full additive} if 
			\begin{equation*}
				\mathfrak{c}_C(\bigvee_{i\in I}\alpha_i)=\bigvee_{i\in I}\mathfrak{c}_C(\alpha_i)
			\end{equation*}
			for all $I\neq \emptyset$ and $C$ such that $\mathscr{P}(C)$ has $I$-indexed suprema;
			\item \emph{idempotent} if 
			$\mathfrak{c}_C\circ \mathfrak{c}_C=\mathfrak{c}_C$
			for all object $C$.
		\end{itemize} 
		
		A \emph{closure doctrine} is a pair $(\mathscr{P}, \mathfrak{c})$ where $\mathscr{P}$ is a primary doctrine and $\mathfrak{c}$ a closure operator on it.
		We say that $(\mathscr{P}, \mathfrak{c})$ is \emph{elementary, existential}, or a \emph{hyperdoctrine},  if $\mathscr{P}$ is. 
	\end{definition}
	
	\begin{example}
Lawvere-Tierney topologies on a topos provide examples of idempotent closure operators on the elementary hyperdoctrine of subobjects \cite{biering2008dialectica,johnstone2002sketches,maclane2012sheaves}.
	\end{example}
	
	\begin{remark} Full additivity does not imply groundedness since we explicitly ask for preservation of suprema indexed on non empty set.
	\end{remark}
		\begin{theoremEnd}{prop}\label{image}
			Let $\mathscr{P}\in \catname{EED}$ be an elementary existential doctrine and  $\mathfrak{c}$ a closure operator on it; then, for any $f:C\rightarrow D$, continuity of $f$ is equivalent to
			$\exists_f\circ \mathfrak{c}_C\leq \mathfrak{c}_D\circ \exists_f$.
		\end{theoremEnd}
	\begin{proof}
		\begin{proofEnd}
			Let's compute:
			\begin{gather*}
			\mathfrak{c}_C\circ \mathscr{P}_f \leq \mathscr{P}_f\circ \exists_{f}\circ \mathfrak{c}_C\circ \mathscr{P}_f  \leq \mathscr{P}_f \circ \mathfrak{c}_D\circ \exists_{f}\circ \mathscr{P}_f \leq \mathscr{P}_f\circ \mathfrak{c}_D
			\\
			\exists_{f}\circ \mathfrak{c}_C  \leq \exists_{f}\circ \mathfrak{c}_C\circ \mathscr{P}_f\circ \exists_{f} \leq \exists_{f}\circ \mathscr{P}_f\circ \mathfrak{c}_D\circ \exists_{f} \leq \mathfrak{c}_D\circ\exists_{f}
			\qedhere\end{gather*}
		\end{proofEnd} \qedhere
	\end{proof}
	If we think of a morphism of (primary, existential, elementary, hyper)doctrines $(\mathscr{F}, \eta):\mathscr{P}\rightarrow \mathscr{Q}$ as a  `translation' of `types' and `predicates' then, when closure operators are available, it is natural to ask for this `translation' to take place in a continuous way.
	\begin{definition}
		A \emph{morphism of closure (elementary, existential, hyper)doctrines} $(\mathscr{F},\eta):(\mathscr{P}, \mathfrak{c})\rightarrow (\mathscr{Q}, \mathfrak{d})$ is a morphism of (elementary, existential, hyper)doctrines $\mathscr{F}:\mathscr{P}\to\mathscr{Q}$ such that $\eta$ is \emph{continuous}, i.e., for all $C$:
		\begin{equation*}
			\mathfrak{d}_{\mathscr{F}(C)}\circ \eta_C \leq \eta_C \circ \mathfrak{c}_C
		\end{equation*}
		We say that $(\mathscr{F}, \eta)$ is \emph{open} if equality holds for all the objects $C$.
		A $2$-cell $\theta:(\mathscr{F},\eta)\rightarrow (\mathscr{G}, \epsilon)$ is defined as in the case of doctrines. In this way we get the $2$-categories $\catname{cPD}$, $\catname{cED}$, $\catname{cHD}$ of closure doctrines, closure existential doctrines, closure hyperdoctrines and the subcategories $\catname{cEPD}$, $\catname{cEED}$, $\catname{cEHD}$ of their elementary variants.
	\end{definition}

	\section{Examples of closure hyperdoctrines}\label{sec:examples}
	
	\subsection{Topological examples}
	As a first class of examples, we introduce three closure hyperdoctrines starting from the usual category $\catname{Top}$ of topological spaces and continuous maps.
	The first one corresponds to the \emph{closure spaces} used in, e.g., \cite{ciancia2014specifying,ciancia2016spatial,galton2003generalized}.
	\begin{definition}\label{def:pretopspaces}
		The category $\catname{PrTop}$ of  \emph{pretopological spaces} (or  \emph{closure spaces}) is the category in which:
		\begin{itemize}
			\item objects are pairs $(X,\mathfrak{c})$ of a set $X$ and a monotone function $\mathfrak{c}:\mathcal{P}(X)\rightarrow\mathcal{P}(X)$ such that $1_{\mathcal{P}(X)}\leq \mathfrak{c}$
			and $\mathfrak{c}$ preserves finite (even empty) suprema;
			
			\item an arrow $f:(X,\mathfrak{c}_X)\rightarrow (Y,\mathfrak{c}_Y)$ is a function $f:X\rightarrow Y$ such that $f^{-1}:(\mathscr{P}(Y), \mathfrak{c}_Y)\rightarrow (\mathscr{P}(X), \mathfrak{c}_X)$ is continuous.
		\end{itemize}
	\end{definition}
	
	Another example is given by so called \emph{convergence spaces} (cfr. \cite{dikranjan2013categorical}).	
	\begin{definition}
		For any set $X$ let $\catname{Fil}(X)$ be the set of proper filters (i.e., $\emptyset$ is not among them) on it.
		The category $\catname{FC}$ of \emph{filter convergence spaces} is the category in which: \begin{itemize}
			\item an object is a pair $(X,q_X)$ given by a set $X$ and a function
			$q_X:X\rightarrow \mathcal{P}(\catname{Fil}(X))$
			such that,  for any $x\in X$,
			$q_X(x)$ is upward closed and $\dot{x}:=\{A\subset X \mid x\in A\}$ belongs to $q_X(x)$. 
			\item an arrow $f:(X,q_x)\rightarrow (Y,q_Y)$ is a function $f:X\rightarrow Y$ such that the filter $f(F)$ generated by the images of $F$'s elements belongs to $q_Y(f(x))$ whenever $F\in q_X(x)$.
		\end{itemize}
	\end{definition}
	
	\begin{proposition}
		The obvious forgetful functors from $\catname{Top}$, $\catname{PrTop}$ and $\catname{FC}$ to $\catname{Set}$ preserves finite products.
	\end{proposition}
	\begin{proof}
		For $\catname{Top}$ it is clear, for the other two categories see \cite[Ch.3]{dikranjan2013categorical}.
	\end{proof}
	By \cref{trans} and the previous one, we have three elementary hyperdoctrines
	\begin{equation*}
	\mathscr{P}^t:\catname{Top}^{op}\rightarrow \catname{HA}\quad \mathscr{P}^p:\catname{PrTop}^{op}\rightarrow \catname{HA}\quad \mathscr{P}^f:\catname{FC}^{op}\rightarrow \catname{HA}
	\end{equation*}
	which we now endow with closure operators.
	
	\begin{definition}\label{def:topclosures}
		We define the following closure operators:
		\begin{enumerate}
			\item the \emph{Kuratowski closure operator} $k=\{k_{(X, \theta)}\}_{(X,\theta)\in \catname{Ob}(\catname{Top})}$ on $\mathscr{P}^t$ where $k_{(X, \theta)}$ is the closure operator associated with the topology $\theta$;
			\item the \emph{\v{C}ech closure operator} $c=\{c_{(X, \mathfrak{c})}\}_{(X,\mathfrak{c})\in \catname{Ob}(\catname{PrTop})}$ on $\mathscr{P}^p$ where $c_{(X, \mathfrak{c})}$ is just $\mathfrak{c}$;
			\item the \emph{Kat\v{e}tov closure operator} $\mathfrak{k}=\{\mathfrak{k}_{(X,q_X)}\}_{(X,\mathfrak{q_X})\in \catname{Ob}(\catname{FC})}$ on $\mathscr{P}^f$ where
			\begin{align*}
			\mathfrak{k}_{(X,q_X)}:\mathcal{P}(X) & \rightarrow \mathcal{P}(X)\\
			A & \mapsto \{x\in X\mid \exists F\in q_X(x).A\in F \}
			\end{align*}
		\end{enumerate}
	\end{definition}

		\begin{theoremEnd}{prop}[\cite{dikranjan2013categorical}, chapter $3$]\label{prtop}
			\begin{enumerate}
				\item $k$, $c$ and $\mathfrak{k}$ are grounded and additive closure operators, moreover $k$ is idempotent.
				\item There exists a sequence of inclusion functors 
				$\catname{Top}\xrightarrow{\mathscr{i}} \catname{PrTop}\xrightarrow{\mathscr{j}} \catname{FC}$
				each of which has a left adjoint.
				\item 	We have a sequence $(\mathscr{P}^t, k)\xrightarrow{(\mathscr{i}, \eta)}	(\mathscr{P}^p, c)	\xrightarrow{(\mathscr{j}, \epsilon)}	(\mathscr{P}^f, \mathfrak{k})$
				of morphisms in $\catname{cEHD}$ where $\eta$ and $\epsilon$ have identities as components.
			\end{enumerate}
		\end{theoremEnd}

	\begin{proof}
		\begin{proofEnd}%
			\begin{enumerate}
				\item For $k$ and $c$ the proposition is obvious, let us examine $\mathfrak{k}$:
				since $\dot{x}\in q_X(x)$ then $A\subset \mathfrak{k}_X(A)$, if $A\subset B$ then any filters that contains the former contains the latter too and this implies monotonicity, groundedness follows from the fact that $\emptyset$ does not belong to any proper filter, for additivity we can complete any filter $\mathcal{F}$ to which $A\cup B$ belong to an ultrafilter $\mathcal{U}$ that belongs to $q_X(x)$ since the latter is upward closed, either $A$ or $B$ must belong to $\mathscr{U}$ and we are done.
				\item $\mathscr{i}$ sends a topological space to the pretopological space given by the closure operator associate to its topology, $\mathscr{j}$ sends $(X, \mathfrak{c})$ to $(X,q^{\mathfrak{c}}_X)$ where
				\begin{align*}
				q^{\mathfrak{c}}_X:X & \rightarrow \mathcal{P}(\catname{Fil}(X))\\
				x & \mapsto \{\mathcal{F}\in\catname{Fil}(X) \mid \mathcal{V}_x\subset \mathcal{F} \}
				\end{align*}
				where $\mathcal{V}_x:=\{S \subset X \mid x\notin \mathfrak{c}(X\smallsetminus S)\}$.
				For the left adjoints see \cite{dikranjan2013categorical}.
				\item This is obvious.\qedhere
			\end{enumerate}
		\end{proofEnd} \qedhere
	\end{proof}
	
	For many other examples of closure operators on topological spaces we refer the interested reader to \cite{dikranjan2013categorical}.

	\subsection{Algebraic examples}	

		\begin{proposition}{prop}
			Let $\catname{Grp}$ be the category of groups and $\catname{CRing}$ that of commutative, unital rings (where we require that $f(1_A)=1_B$ for any $f:A\rightarrow B$).
			Then, $\sub{Grp}$ and $\sub{CRing}$ are elementary existential doctrines.
		\end{proposition}
	\begin{proof}
		This follows at once from theorem \ref{reg}. \qedhere 
	\end{proof}
	\begin{remark}
		Notice that, even if $\sub{Grp}(G)$ and $\sub{CRing}(A)$ admit finite suprema for any group $G$ or commutative ring $A$ with unity, preimages do not preserve them in general: for instance they do not preserve the bottom subobject. Then $\sub{Grp}$ or $\sub{CRing}$ cannot be universal doctrines.
	\end{remark}

	The following examples are taken from \cite{dikranjan2013categorical}.		
	\begin{definition}[Groups] The \emph{normal closure} on a group $G$ is given by
		\begin{align*}
		\nu_G:\sub{Grp}(G) & \rightarrow \sub{Grp}(G)\\
		H & \mapsto \bigcap\{N\leq G\mid H\leq N \trianglelefteq G\}
		\end{align*}
		where we have chosen the image of a monomorphism as a canonical representative of it.
	\end{definition}
		\begin{theoremEnd}{prop}
			The family previous defined forms a closure operators $\nu$ on $\sub{Grp}$that is idempotent, fully additive and grounded.
		\end{theoremEnd}
		\begin{proof}
		\begin{proofEnd}
			Since the preimage of a normal subgroup is normal we have that the $\nu$ actually exists as a closure operator. The three poperties of it follow immediately by the fact that $\{0\}$ is normal and so are the arbitrary intersections or sums of normal subgroups.\qedhere
		\end{proofEnd} \qedhere
	\end{proof}
	\begin{definition}[Rings]
		Let $A$ be a unital commutative ring and $B$ a subring, we define $\integ{A}(B)$ to be the \emph{integral closure of $B$}:
		\begin{equation*}
		\integ{A}(B):=\{a\in A\mid p(a)=0 \ \text{for some} \ p\in B[x]\}
		\end{equation*}
		Again we are denoting a subobject by the image of any representative of it.
	\end{definition}
		\begin{theoremEnd}{prop} For any $A$ $\integ{A}$ is a function $\sub{CRing}(A)\rightarrow\sub{CRing}(A)$, moreover the family of this functions forms an idempotent closure operator $\integ{}$. 
		\end{theoremEnd}	
	\begin{proof}
		\begin{proofEnd}
			To show that $\integ{A}(B)$ is a subring of $A$ and idempotency we refer to \cite[Cor.~5.3, 5.5]{atiyah2018introduction}.
			Let us show that $\integ{}$ is actually a closure operator. Consider $f:A\rightarrow B$ and $C$ a subring of $B$, let $a\in A$ such that $p(a)=0$ for some $p\in f^{-1}(C)[X]$ with coefficients $\{p_i\}_{i=0}^{\deg(p)}$, then $q(f(a))=0$ where $q\in C[X]$ has coefficients $\{f(p_i)\}_{i=0}^{\deg(p)}$ and we are done.\qedhere
		\end{proofEnd} \qedhere
	\end{proof}
	\subsection{A representable example}\label{sec:reprex}
	\begin{theorem}
		For any complete Heyting algebra $H$, the functor $\catname{Set}(-, H):\catname{Set}^{op}\to\catname{HA}$ is an elementary hyperdoctrine.
	\end{theorem}
\begin{proof}
	This is example $2.2$ in \cite{pitts2002tripos}. Explicitly we have
	\begin{equation*}
	\begin{aligned}
	\exists_{\pi_X}(f):X & \rightarrow H\\
	x & \mapsto \bigvee_{y\in Y}f(x,y)
	\end{aligned}\qquad 
	\begin{aligned}
	\forall_{\pi_X}(f):X & \rightarrow H\\
	x & \mapsto \bigwedge_{y\in Y}f(x,y)
	\end{aligned}\qquad
	\begin{aligned}
		\delta_X:X\times X&\rightarrow [0,1]\\(x,y)&\mapsto \begin{cases}
		0 & x\neq y\\
		1 &x=y
		\end{cases}
	\end{aligned}
\end{equation*}
\end{proof}

	\begin{corollary}
		$\catname{Set}(-, [0,1]):\catname{Set}^{op}\to\catname{HA}$ is an elementary hyperdoctrine on $\catname{Set}$.
	\end{corollary}
	\begin{definition}
		For any fixed  real $\epsilon \geq 0$, and any set $X$ we define, for an $f:X\rightarrow [0,1]$ we define
		\begin{gather*}
		\begin{aligned}
		\mathfrak{c}_{X, \epsilon}(f):X & \rightarrow [0,1]\\
		x & \mapsto f(x)\dot{+}\epsilon
		\end{aligned}
		\quad \text{where} \quad
		\begin{aligned}
		\dot{+}:[0,1]\times [0,1]&\rightarrow [0,1]\\
		(t,s)& \mapsto \max(t+s, 1)
		\end{aligned}
		\end{gather*}
		In this way we get a function
		\begin{align*}
		\mathfrak{c}_{X, \epsilon}: \catname{Set}(X, [0,1]) & \rightarrow \catname{Set}(X, [0,1])\\
		f & \mapsto 	\mathfrak{c}_{X, \epsilon}(f)
		\end{align*}
	\end{definition}
		\begin{theoremEnd}{prop}
			For any $\epsilon \geq 0$, the collection $\mathfrak{c}_\epsilon$ of all the functions $\mathfrak{c}_{X, \epsilon}$ is a closure operator.
		\end{theoremEnd}
	\begin{proof}
		\begin{proofEnd}
			Clearly $f\leq\mathfrak{c}_{X,\epsilon}(f)$ for any $f:X\rightarrow [0,1]$, monotonicity is clear, let's check continuity of any function $g:X\rightarrow Y$:
			\begin{align*}
			\mathfrak{c}_{X,\epsilon}(f\circ g)(x)
			&=	(f\circ g)(x)\dot{+}\epsilon \\
			&= f(g(x))\dot{+}\epsilon \\
			&=\mathfrak{c}_{x,\epsilon}(f)(g(x))\\&=(\mathfrak{c}_{x,\epsilon}(f)\circ g)(x)
			\qedhere
			\end{align*}
		\end{proofEnd} \qedhere
	\end{proof}
	\begin{remark}
		$\mathfrak{c}_{\epsilon}$ is not grounded if $\epsilon \neq 0$ (in that case it reduces to the discrete closure operator) but it is additive.\end{remark}

	\subsection{Fuzzy sets}
	We can refine the previous example considering \emph{fuzzy sets}.
	\begin{definition}
		The category \emph{$\catname{Fzs}$ of fuzzy sets} has:
		\begin{itemize}
			\item pairs $(A, \alpha)$ with $\alpha:A\rightarrow [0,1]$ as objects;
			\item as arrows $f:(A, \alpha)\rightarrow (B, \beta)$ functions $f:A\rightarrow B$ such that 
			$\alpha(x)\leq \beta(f(x))$.
		\end{itemize}
	\end{definition}
	\begin{definition}
		A \emph{fuzzy subset} of $(A, \alpha)$ is a function $\xi:A\rightarrow [0,1]$ such that
		$\xi(x)\leq \alpha(x)$	for all $x\in A$.
	\end{definition}
	Let us summarize some results about $\catname{Fzs}$.
	\begin{proposition}
		\begin{enumerate}
			\item $\catname{Fzs}$ is a quasitopos; 
			\item there exists a proper and stable factorization system given by strong monomorphisms and epimorphisms;
			\item fuzzy subsets of $(A,\alpha)$ correspond to equivalence of strong monomorphisms of codomain $(A, \alpha)$;
			\item the functor 
			\begin{align*}
			\catname{Fzs}^{op}& \rightarrow \catname{HA}\\
			\functor[l]{(A, \alpha)}{f}{(B, \beta)}
			& \functoropmapsto
			\rfunctorop{\fuz(A, \alpha)}{f^*}{\fuz(B, \beta)}
			\end{align*}
			where $\fuz(A, \alpha)$ is the set of fuzzy subsets of $(A, \alpha)$ and 
			\begin{align*}
			f^*(\xi):A&\rightarrow [0,1]\\
			x&\mapsto \alpha(x)\wedge \xi(f(x))
			\end{align*}
			for any $\xi \in \fuz(B, \beta)$,
			is an elementary hyperdoctrine.
		\end{enumerate}
	\end{proposition}
	\begin{proof}
		See \cite[Ch.~8]{oswald1991lecture}. Explicitly the hyperdoctrine structure is given by:
		\begin{gather*}
		\begin{aligned}
		\exists_f(\xi):B&\rightarrow [0,1]\\
		y &\mapsto \bigvee_{x\in f^{-1}(y)}\xi(x)
		\end{aligned}\qquad
		\begin{aligned}
		\forall_f(\xi):B&\rightarrow [0,1]\\
		y &\mapsto \beta(y) \wedge \bigwedge_{x\in f^{-1}(y)}(\alpha(x)\Rightarrow \xi(x))
		\end{aligned}
		\end{gather*}
		for any $f:(A, \alpha)\rightarrow (B, \beta)$ and $\xi \in \fuz(A, \alpha)$.
		\qedhere
	\end{proof}
	\begin{remark}
		Implication in $[0,1]$ is given by:
		\begin{equation*}
		t \Rightarrow s = \begin{cases}
		1 &t\leq s\\
		s &s < t
		\end{cases}
		\end{equation*}
		Moreover the fibered equality for a fuzzy set $(A,\alpha)$ must be $\exists_{\Delta_{(A,\alpha)}}(\alpha)$, i.e.:
		\begin{align*}
		\delta_{(A, \alpha)}:A\times A & \rightarrow [0,1]\\
		(x,y) & \mapsto \begin{cases}
		\alpha (x) &x=y\\
		0 &x\neq y
		\end{cases}
		\end{align*}
		Notice that in $\catname{Fzs}$,  $(A, \alpha)\times (B, \beta)$ is $(A\times B, \alpha \wedge \beta)$.
	\end{remark}
		\begin{theoremEnd}{prop} 
			Let $\mathscr{E}=\{\epsilon_{(A, \alpha)}\}_{(A,\alpha)\in \catname{Ob(Fzs)}}$ be a family of functions $\epsilon_{(A, \alpha)}:(A,\alpha)\rightarrow [0,1]$ such that, for any $f:(A,\alpha)\rightarrow (B,\beta)$
			\begin{equation*}
			\epsilon_{(A, \alpha)}(x)\leq \epsilon_{(B,\beta)}(f(x))
			\end{equation*}
			then, we get an additive closure operator on $\fuz$ defined as follows:
			\begin{align*}
			\mathfrak{c}^{\mathscr{E}}_{(A,\alpha)}: \fuz(A,\alpha)&\rightarrow \fuz(A,\alpha)\\
			\xi &\mapsto (\xi+\epsilon_{(A, \alpha)})\wedge \alpha
			\end{align*}
		\end{theoremEnd}
	\begin{proof}
		\begin{proofEnd}
			We have to show continuity of all arrows $f:(A,\alpha)\rightarrow (B,\beta)$. Let $\xi\in (B,\beta)$ and $x\in A$, we have four cases:
			\begin{enumerate}
				\item $f^*(\xi)(x) + \epsilon_{(A,\alpha)}(x)< \alpha(x)$ and $\xi(x) + \epsilon_{(B,\beta)}(x)< \beta(x)$. 
				\begin{align*}
				(\mathfrak{c}^{\mathscr{E}}_{(A,\alpha)}(f^*(\xi)))(x)&=(f^*(\xi) + \epsilon_{(A,\alpha)})(x)\\&=(\alpha(x)\wedge \xi(f(x)))+\epsilon_{(A, \alpha)}(x)\\&=\alpha(x)\wedge (\xi(f(x))+\epsilon_{(A, \alpha)}(x))\\&\leq \alpha(x)\wedge (\xi(f(x))+\epsilon_{(B,\beta)}(f(x)))\\&=f^*(\mathfrak{c}^{\mathscr{E}}_{(B,\beta)}(\xi))(x)
				\end{align*}
				\item $f^*(\xi)(x) + \epsilon_{(A,\alpha)}(x)< \alpha(x)$ and $\xi(f(x)) + \epsilon_{(B,\beta)}(f(x))\geq  \beta(f(x))$. Notice that $\alpha(x)\leq\beta(f(x))$ so
				\begin{equation*}
				f^*(\mathfrak{c}^{\mathscr{E}}_{(B,\beta)}(\xi))(x)=\alpha(x)
				\end{equation*}
				from which:
				\begin{align*}
				(\mathfrak{c}^{\mathscr{E}}_{(A,\alpha)}(f^*(\xi)))(x)&=(f^*(\xi) + \epsilon_{(A,\alpha)})(x)\\&=(\alpha(x)\wedge \xi(f(x)))+\epsilon_{(A, \alpha)}(x)\\&=\alpha(x)\wedge  (\xi(f(x))+\epsilon_{(A, \alpha)}(x))\\&=\alpha(x)\\&=f^*(\mathfrak{c}^{\mathscr{E}}_{(B,\beta)}(\xi))(x)
				\end{align*}
				\item $f^*(\xi)(x) + \epsilon_{(A,\alpha)}(x)\geq \alpha(x)$ and $\xi(x) + \epsilon_{(B,\beta)}(x)< \beta(x)$.
				\begin{align*}
				(\mathfrak{c}^{\mathscr{E}}_{(A,\alpha)}(f^*(\xi)))(x)&=\alpha(x)\\&=\alpha(x)\wedge  (\xi(f(x))+\epsilon_{(A, \alpha)}(x))\\&\leq \alpha(x)\wedge (\xi(f(x))+\epsilon_{(B,\beta)}(f(x)))\\&=f^*(\mathfrak{c}^{\mathscr{E}}_{(B,\beta)}(\xi))(x)
				\end{align*}
				\item $f^*(\xi)(x) + \epsilon_{(A,\alpha)}(x)\geq \alpha(x)$ and $\xi(x) + \epsilon_{(B,\beta)}(x)\geq \beta(x)$. \begin{align*}
				(\mathfrak{c}^{\mathscr{E}}_{(A,\alpha)}(f^*(\xi)))(x)&=\alpha(x)\\&=\alpha(x)\wedge\beta(f(x))\\&=f^*(\mathfrak{c}^{\mathscr{E}}_{(B,\beta)}(\xi))(x)
				\end{align*}
				We are left with additivity, but this follows immediately since, for $\xi$ and $\zeta\in \fuz(A,\alpha)$ and $x\in A$, 
				$(\xi \vee \zeta)(x)$ is $\xi(x)$ or $\zeta(x)$.
				\qedhere
			\end{enumerate}
		\end{proofEnd} \qedhere
	\end{proof}
	\begin{remark}
		$\mathfrak{c}^{\mathscr{E}}$ is not grounded in general.
	\end{remark}	
	
	The condition on the elements of $\mathscr{E}$ is very restrictive. In fact, it can be eased restricting to a suitable subclass of arrows and using the following lemma.
	\begin{lemma}\label{rest}
		Let $\mathscr{P}:\catname{C}^{op}\rightarrow \catname{InfSL}$ be a  doctrine, and $\mathfrak{c}=\{\mathfrak{c}_C:\mathscr{P}(C)\rightarrow \mathscr{P}(C)\}_{C\in \catname{Ob(C)}}$ be a family of monotone and inflationary operators.
		Let $\mathscr{A}$ be a (possibly large) family of \catname{C}-arrows such that:
		\begin{itemize}
			\item $\mathscr{A}$ is closed under composition;
			\item if $f\in \mathscr{A}$ then $1_{\dom(A)}$ and $1_{\cod(A)}$ are in $\mathscr{A}$;
			\item $f:C\rightarrow D$ in $\mathscr{A}$ implies
			$\mathfrak{c}_C\circ \mathscr{P}_f\leq \mathscr{P}_f \circ \mathfrak{c}_D$.
		\end{itemize}
		Then $\mathscr{P}$ induces a doctrine $\mathscr{P}^{\mathscr{A}}$ on the subcategory $\catname{C}_{\mathscr{A}}$ induced by $\mathscr{A}$ for which $\mathfrak{c}=\{\mathfrak{c}_C\}_{C\in \catname{Ob(C}_\mathscr{A}\catname{)}}$ is a closure operator.
		Moreover, if for all $f,g$ in $\mathscr{A}$ also $(f, g)$ and the projections from $\cod(f)\times \cod(g)$ are in $\mathscr{A}$, 
		then $\mathscr{P}^{\mathscr{A}}$ is existential, elementary or an hyperdoctrine if $\mathscr{P}$ is. 
	\end{lemma}
	\begin{proof}
		This is almost tautological since the condition on $\mathscr{A}$ guarantee that the inclusion functor $\catname{C}_{\mathscr{A}}$ preserves limits and we can use \cref{trans}.\qedhere
	\end{proof}
	
	\subsection{Coalgebraic examples}
	\begin{definition}[\cite{jacobs2017introduction,kupke2011coalgebraic}]
		Let $\catname{C}$ be a category with finite products and $\mathscr{F}:\catname{C}\rightarrow \catname{C}$ an endofunctor. 
		The category $\coalg{F}$ of \emph{coalgebras for $\mathscr{F}$} has
		\begin{itemize}
			\item arrows $\gamma_C:C\rightarrow \mathscr{F}(C)$ as objects;
			\item arrows $f:C\rightarrow D$ such that $ \gamma_D\circ f = \mathscr{F}(f)\circ\gamma_C$
			as morphisms $f:\gamma_C\rightarrow \gamma_D$.
		\end{itemize}
	\end{definition}
	Notice that in general $\coalg{F}$ is not complete and products in it can be very different from products in $\catname{C}$ \cite{gumm2001products}, so it does not make much sense to look for an existential doctrine on it. However, for $\catname{Set}$-based coalgebras we get a primary doctrine $\mathscr{P}^c:\coalg{F}^{op}\rightarrow \catname{InfSL}$ composing the contravariant power object $\mathscr{P}:\catname{Set}^{op}\rightarrow \catname{InfSL}$ with the opposite of the obvious forgetful functor $\coalg{F}\rightarrow \catname{Set}$.
	
	\begin{definition} Let $\mathscr{F}:\catname{C}\rightarrow \catname{C}$ be a functor and $\mathscr{P}$ a primary doctrine on $\catname{C}$.
		A \emph{predicate lifting} is a natural transformation $\Box: \mathscr{U}\circ\mathscr{P} \rightarrow \mathscr{U}\circ \mathscr{P}\circ \mathscr{F}^{op}$ where $\mathscr{U}$ is the forgetful functor $\catname{InfSL}\rightarrow \catname{Poset}$.
	\end{definition}
	
	Let $\Box$ be a predicate lifting.
	We define two closure operators on $\mathscr{P}^c$.
	\begin{enumerate}
		\item For any coalgebra $\gamma_X:X\rightarrow \mathscr{F}(X)$, 
		notice that $\mathscr{P}^c(\gamma_X) = \mathscr{P}(X)$; hence we can define
		\vspace{-1.5ex}
		\begin{align*}
		\diam{X}:\mathscr{P}(X)&\rightarrow \mathscr{P}(X)\\
		\alpha &\mapsto \alpha \vee \mathscr{P}_{\gamma_X}( \Box_X (\alpha))
		\end{align*}
		\item Suppose that $\mathscr{P}$ admits arbitrary meets; for $\gamma_X:X\rightarrow \mathscr{F}(X)$ and $\alpha\in \mathscr{P}(X)$ we define
		\begin{align*}
		\mathcal{N}_{\gamma_X}(\alpha) & :=\{\beta \in \mathscr{P}(X)\mid \alpha \leq 	\mathscr{P}_{\gamma_X}(\Box_X(\beta))\} \\
		\mathsf{s}_{\gamma_X}(\alpha) & := \bigwedge_{\beta\in\mathcal{N}_{\gamma_X}(\alpha)}\beta \\
		\scat{X}&:\mathscr{P}(X) \rightarrow \mathscr{P}(X)\\
		          & \hspace{9mm} \alpha  \mapsto \alpha\vee \mathsf{s}_{\gamma_X}(\alpha)
		\end{align*}
	\end{enumerate}	
		\begin{theoremEnd}{lem} Let $\mathscr{F}:\catname{C}\rightarrow \catname{C}$ be a functor and $\Box$ a predicate lifting, then:
			\begin{enumerate}
				\item $\{\diam{X}\}_{\gamma_X \in \catname{Ob}(\coalg{F})}$ defines a closure operator $\diamo$ on $\mathscr{P}^c$.
				\item  $\mathsf{s}_{\gamma_X}(\alpha)$ is the minimum of $\mathcal{N}_{\gamma_X}(\alpha)$ whenever $\mathscr{P}$ has arbitrary meets and, for any coalgebra $\gamma_X:X\rightarrow \mathscr{F}(X)$, $\mathscr{P}_{\gamma_X}$ and $\Box_X$ commute with them;
				\item in the hypothesis above if $\mathscr{P}_f$ commutes with arbitrary meets for all arrows $f$ then  $\{\scat{X}\}_{\gamma_X \in \catname{Ob}(\coalg{F})}$ defines a closure operators $\scato$ on $\mathscr{P}^c$.				
			\end{enumerate}
		\end{theoremEnd}
	\begin{proof}
		\begin{proofEnd}
			\begin{enumerate}
				\item Clearly $\alpha\leq 
				\diam{X}(\alpha)$; if $\alpha \leq \beta$ we have that
				\begin{equation*}
				\mathscr{P}_{\gamma_{X}}(\Box_X(\alpha))\leq \mathscr{P}_{\gamma_{X}}(\Box_X(\beta)) 
				\end{equation*}  
				from which monotonicity follows; for $f$ an arrow between $\gamma_X:X\rightarrow \mathscr{F}(X)$ and $\gamma_Y:Y\rightarrow \mathscr{F}(Y)$, we have a commutative diagram
				\begin{center}
					\begin{tikzpicture}
					\node(A)at (0,0) {$X$};
					\node(B)at (3,0) {$Y$};
					\node(C)at (0,-1.5) {$\mathscr{F}(X)$};
					\node(D)at (3,-1.5) {$\mathscr{F}(Y)$};
					\draw[->](A)--(C) node[pos=0.5, left]{$\gamma_X$};
					\draw[->](B)--(D) node[pos=0.5, right]{$\gamma_Y$};
					\draw[->](A)--(B) node[pos=0.5, above]{$f$};
					\draw[->](C)--(D) node[pos=0.5, below]{$\mathscr{F}(f)$};
					\end{tikzpicture}
				\end{center} and computing we get the thesis:
				\begin{align*}
				\diam{X}(\mathscr{P}_f(\alpha))&=\mathscr{P}_f(\alpha)\vee \mathscr{P}_{\gamma_X}(\Box_X(\mathscr{P}_f(\alpha)))\\&=\mathscr{P}_f(\alpha)\vee \mathscr{P}_{\gamma_X}(\mathscr{P}_{\mathscr{F}(f)}(\Box_Y(\alpha)))\\&=\mathscr{P}_f(\alpha)\vee \mathscr{P}_f(\mathscr{P}_{\gamma_Y}(\Box_Y(\alpha)))\\&=
				\mathscr{P}_f(\alpha \vee \mathscr{P}_{\gamma_Y}(\Box_Y(\alpha)))\\&=\mathscr{P}_f(\diam{Y}(\alpha))
				\end{align*}

				\item By hypothesis:
				\begin{align*}
				\alpha &\leq \bigwedge_{\beta \in \mathcal{N}_{\gamma_X}(\alpha)}\mathscr{P}_{\gamma_X}(\Box_X(\beta))\\&=
				\mathscr{P}_{\gamma_X}(\bigwedge_{\beta \in \mathcal{N}_{\gamma_X}(\alpha)}\Box_X(\beta))\\&=			\mathscr{P}_{\gamma_X}(\Box_X(\bigwedge_{\beta \in \mathcal{N}_{\gamma_X}(\alpha)}\beta))\\&=\mathscr{P}_{\gamma_X}(\Box_X(\mathsf{s}_{\gamma_X}(\alpha)))
				\end{align*}
				
				\item The inequality $\alpha\leq 
				\scat{X}(\alpha)$ follows at once, if $\alpha \leq \beta$ we have $\mathscr{P}_{\gamma_{X}}(\Box_X(\alpha))$ as in the first point but this implies that 
				$\mathcal{N}_{\gamma_X}(\beta)\subset \mathcal{N}_{\gamma_X}(\alpha)$. 
				Hence, $\bigwedge_{\theta \in \mathcal{N}_{\gamma_X}(\alpha)}\theta\leq \bigwedge_{\theta \in \mathcal{N}_{\gamma_X}(\beta)}\theta$, 
				from which we deduce the monotonicity of $\scat{X}$. Let now $f:X\rightarrow Y$ be an arrow such that		\begin{center}
					\begin{tikzpicture}
					\node(A)at (0,0) {$X$};
					\node(B)at (3,0) {$Y$};
					\node(C)at (0,-1.5) {$\mathscr{F}(X)$};
					\node(D)at (3,-1.5) {$\mathscr{F}(Y)$};
					\draw[->](A)--(C) node[pos=0.5, left]{$\gamma_X$};
					\draw[->](B)--(D) node[pos=0.5, right]{$\gamma_Y$};
					\draw[->](A)--(B) node[pos=0.5, above]{$f$};
					\draw[->](C)--(D) node[pos=0.5, below]{$\mathscr{F}(f)$};
					\end{tikzpicture}
				\end{center}
				commutes, and notice that for all $\theta\in \mathcal{N}_Y(\alpha)$ then
				\begin{align*}
				\mathscr{P}_f(\alpha)&\leq \mathscr{P}_f(\mathscr{P}_{\gamma_Y}(\Box_Y(\theta))) 
				=\mathscr{P}_{\gamma_X}(\mathscr{P}_{\mathscr{F}(f)}(\Box_Y(\theta))) =\mathscr{P}_{\gamma_X}(\Box_X(\mathscr{P}_f(\theta)))
				\end{align*}
				hence $\mathscr{P}_f(\theta)\in \mathcal{N}_X(\mathscr{P}_f(\alpha))$ and thus
				\begin{align*}
				\scat{X}(\mathscr{P}_f(\alpha))& = \mathscr{P}_f(\alpha)\vee \mathsf{s}_{\gamma_X}(\mathscr{P}_f(\alpha))=\mathscr{P}_f(\alpha)\vee \bigwedge_{\beta \in \mathcal{N}_X(\mathscr{P}_f(\alpha))}\beta\\&\leq \mathscr{P}_f(\alpha)\vee \bigwedge_{\beta \in \mathcal{N}_X(\mathscr{P}_f(\alpha))}\beta
				\\&\leq \mathscr{P}_f(\alpha)\vee\bigwedge_{\beta \in \mathcal{N}_Y(\alpha)}\mathscr{P}_f(\beta)\\&\leq \mathscr{P}_f(\alpha)\vee \mathscr{P}_f(\bigwedge_{\beta \in \mathcal{N}_Y(\alpha)}\beta)\\&=\mathscr{P}_f(\alpha\vee \mathsf{s}_{\gamma_Y}(\alpha)) =\mathscr{P}_f(\scat{Y}(\alpha))
				\end{align*}
				and we are done.
				\qedhere
			\end{enumerate}
		\end{proofEnd} \qedhere
	\end{proof}
	
	The previous result provides us with many examples with practical applications.
	\begin{example}[Kripke frames]
		Let $\mathcal{P}:\catname{Set}\rightarrow \catname{Set}$ be the covariant powerset functor, and 
		$\mathscr{P}:\catname{Set}^{op}\rightarrow \catname{InfSL}$ be the controvariant one, seen as primary doctrine.
		We can define a predicate lifting $\Box$ taking as components:
		\begin{align*}
		\Box_X: \mathscr{P}(X) & \rightarrow \mathscr{P}(\mathcal{P}(X))\\
		A & \mapsto \small{\downarrow} A
		\end{align*} 
		where $\small{\downarrow} A$ denotes the set of downward-closed subsets of $A$.
		In this case for any coalgebra $\gamma_X:X\rightarrow \mathcal{P}(X)$ we have
		\begin{align*}
		x\in \gamma_X^{-1}(\Box_X(A)) & \iff \gamma_X(x)\subset A\\
		B\in \mathcal{N}_{\gamma_X}(A) & \iff \gamma_X(a)\subset B \text{ for any } a\in A 
		\end{align*}
		so 
		$\mathsf{s}_{\gamma_X}(A)=\bigcup_{a\in A} \gamma_X(a)$ and $\scat{X}(A)=A\cup 	\bigcup_{a\in A}\gamma_X(a)$.
		
		By this description it is clear that $\scato$ is grounded and fully additive.  
		$\diamo$ is grounded too but it is not even finitely additive: take $4:=\{0,1,2,3\}$ with stuctural map $\gamma_4$ given by 
		\begin{gather*}
		0\mapsto \{3\} \quad
		1\mapsto \{2,3\} \quad
		2\mapsto \{2\} \quad
		3\mapsto \{3\}
		\end{gather*}
		Now take $A:=\{2,3\}$, it is immediate to see that $\diam{4}(A)=4$, on the other hand $\diam{4}(\{2\})=\{2\}$ and $\diam{4}(\{3\})=\{0,3\}$.
		
		\begin{remark}
			In this case, $\diamo$ and $\scato$ meanings (and notation) become clearer: if we think to the value of $\gamma_X(x)$ as the family of points accessible from $x\in X$ then $\diam{X}$ adds to a subset $A$ the set of its \emph{predecessors}, i.e. points from which some  $a\in A$ is accessible, while $\scat{X}$ adds the set of \emph{successors}, i.e. points which are accessible from some point of $A$.
		\end{remark} 
	\end{example}
	
	\begin{example}[Probabilistic frames \cite{avery2016codensity,bacci2015structural,giry1982categorical}]
		Let $\catname{Meas}$ be the category of measurable space and measurable functions; then we can take as primary doctrine $\mathscr{P}$ the functor 
		\begin{equation*}
		\functor[l]{(X,\Omega_X)}{f}{(Y,\Omega_Y)}
		\functormapsto
		\functor[r]{\Omega_X}{f^{-1}}{\Omega_Y}	
		\end{equation*}
		As endofunctor we can take the \emph{Giry monad} $\mathscr{G}:\catname{Meas}\to\catname{Meas}$:
		\begin{itemize}
			\item given an object $(X,\Omega_X)$, $\mathscr{G}(X,\Omega_X)$ is the set
			\begin{equation*}
			\{\mu:\Omega_X\rightarrow [0,1]\mid \mu \text{ is a probabilty measure on } \Omega_X \}
			\end{equation*}
			equipped with the smallest $\sigma$-algebra for which all the \emph{evaluation functions}
			\begin{align*}
			\mathsf{ev}_A:\mathscr{G}(X,\Omega_X)&\rightarrow [0,1]\\
			\mu &\mapsto \mu(A)
			\end{align*} 
			with $A\in \Omega_X$, are Borel-measureable.
			
			\item for a measurable $f:(X,\Omega_X)\rightarrow (Y,\Omega_Y)$,
			\begin{align*}
			\mathscr{G}(f):\mathscr{G}(X,\Omega_X)&\rightarrow \mathscr{G}(Y,\Omega_Y)\\
			\mu& \mapsto \mu \circ f^{-1}
			\end{align*}
		\end{itemize}
		(For the measurability of $\mathscr{G}(f)$ notice that given a Borel subset $L$ of $[0,1]$ and $A\in \Omega_Y$ we have that 
		$\mu \in \mathscr{G}(f)^{-1}(\mathsf{ev}_A(L))\iff \mu \in \mathsf{ev}_{f^{-1}(A)}(L)$)
		
		For a coalgebra $\gamma_{(X,\Omega_X)}$ and $p\in [0,1]$ we define
		\begin{align*}
		\Box_{(X,\Omega_X), p}:\Omega_X&\rightarrow \mathcal{P}(\mathscr{G}(X))\\
		A & \mapsto \{\mu \in \mathscr{G}(X,\Omega_X) \mid\mu(A)\geq p \}
		\end{align*}
		notice that the set on the right is $\mathsf{ev}_A^{-1}([p,1])$ and so $\Box_{(X,\Omega_X), p}$ is well defined. In this situation we have
		\begin{equation*}
		\diam{(X,\Omega_X)}(A):=A\cup \{x\in X\mid p\leq \gamma_{(X,\Omega_X)}(x)(A)\}
		\end{equation*}
	\end{example}
	\begin{remark}
		If we think of  a coalgebra $\gamma_{(X,\Omega_X)}$ as describing how likely is a transition from a state to the various $A\in \Omega_X$ then, given a $p\in [0,1]$, $\diam{(X,\Omega_X)}(A)$ is the set of points which access  $A$  with probability at least  $p$.
	\end{remark}

	\section{Logics for Closure Operators} \label{sec:slcs}
	In this section, we provide a sound and complete logic for closure hyperdoctrines. 
	This logic is a (first order) version of Spatial Logic for Closure Spaces (SLCS) \cite{ciancia2016spatial}, although with a slightly different presentation.
	
	\subsection{Syntax and derivation rules}
	
	We briefly recall the categorical presentation of signatures, as in \cite{jacobs1999categorical}.
	\begin{definition}\label{def:signature}
		A \emph{signature} $\Sigma$ is a triple $(\abs{\Sigma}, \Gamma, \Pi )$ where
		\begin{itemize}
			\item $\abs{\Sigma}$ is a set, called the set of \emph{basic types};
			\item $\Gamma$ is a functor\footnote{$\abs{\Sigma}$ and $\abs{\Sigma}^\star$ are viewed here as discrete categories.} $\abs{\Sigma}^\star \times \abs{\Sigma}\rightarrow \catname{Sets}$. We will call \emph{function symbol} an element $f$ of $\Gamma((\sigma_1,\dots,\sigma_n), \sigma_{n+1})$ and we will write $f:\sigma,\dots,\sigma_n\rightarrow, \sigma_{n+1}$;
			\item $\Pi$ is a functor $\abs{\Sigma}^\star\rightarrow \catname{Set}$,  we will call \emph{predicate symbol} an element $P$ of $\Pi(\sigma_1,\dots,\sigma_n)$ and we will write $P:\sigma_1,\dots,\sigma_n$.
		\end{itemize}
		A morphism of signatures $\phi:\Sigma_1\rightarrow \Sigma_2$ is a triple $(\phi_1,\phi_2,\phi_3)$ such that
		\begin{itemize}
			\item $\phi_1$ is a function $\abs{\Sigma_1}\rightarrow \abs{\Sigma_2}$;
			\item $\phi_2$ is a natural transformation $\Gamma_1\rightarrow \Gamma_2\circ (\phi_1^{\star}\times \phi_1)$;
			\item $\phi_3$ is a natural transformation $\Pi_1\rightarrow \Pi_2\circ \phi_1^{\star}$.
		\end{itemize}
		For any $\sigma\in \abs{\Sigma}$ we fix an countably infinite set $X_\sigma$ of \emph{variables}; definition of terms is straightforward (\cite{jacobs1999categorical}).
	\end{definition}
	
	\iffalse
	\begin{proposition}\label{pb}
		Signature and their morphisms with componentwise composition form a category $\catname{SignPred}$, moreover this category is the pullback of the codomain fibration $\cod:\catname{Set}^{\catname{2}}\rightarrow \catname{Set}$
		along the functor $\catname{Set}\rightarrow\catname{Set}$
		\begin{equation*}
		\functor[l]{A}{f}{B}
		\functormapsto
		\functor[r]{(A^\star\times A)\sqcup A^\star}{(f^\star \times f)\sqcup f^\star}{(B^\star\times B)\sqcup B^\star}
		\end{equation*}
	\end{proposition}
	\begin{proof}
		This is straightforward (\cite{jacobs1999categorical} definition $4.1.1$). \qedhere
	\end{proof}
	\fi	
	
	\begin{definition}
		Given a signature $\Sigma$, its \emph{classifying category} $\class{\Sigma}$ is such that
		\begin{itemize}
			\item objects are contexts;
			\item Given $\Gamma:=[x_i:\sigma_i]_{i=1}^n$ and $\Delta=[y_i:\tau_i]_{i=1}^m$ an arrow $\Gamma\rightarrow \Delta$ is a $m$-uple of terms $(T_1,...,T_m)$ such that $\Gamma \vdash T_i:\tau_i$ for any $i$;
			\item composition is given by substitution.
		\end{itemize}
	\end{definition}
		\begin{theoremEnd}{prop}
			$\class{\Sigma}$ is a category with finite products for any signature $\Sigma$.
		\end{theoremEnd}
	\begin{proof}
		\begin{proofEnd}Associativity of composition and the fact that $(x_1,...,x_n)$ is the identity for $[x_i:\sigma_i]_{i=1}^n$ follows from a straightforward computation. The empty context is clearly terminal while, given two contexts $\Gamma:=[x_i:\sigma_i]_{i=1}^n$ and $\Delta=[y_i:\tau_i]_{i=1}^m$ we can take their concatenation as a product $\Gamma \times \Delta$, the universal property follows immediately.\qedhere 
		\end{proofEnd} \qedhere
	\end{proof}

	Now we can introduce the rules for context and closure operators of the Spatial Logic for Closure Spaces, over any given signature. 
	
	As usual, we denote by $\Gamma \vdash t : \tau$ the judgment ``$t$ has type $\tau$ in context $\Gamma$'', and by $\Gamma \vdash \phi : \propo$ the judgment ``$\phi$ is a well-formed formula in context $\Gamma$''.
	\begin{definition}\label{def:slcssynt}
		The rules for contexts and well-formed formulae for the closure operators for a signature $\Sigma$ are the usual ones for a first order signature  (see \cite{jacobs1999categorical}) plus:
		\begin{equation*}
		\inferrule*[right=$\mathcal{C}$-F]{\Gamma \vdash \phi:\propo}{\Gamma \vdash \mathcal{C}(\phi):\propo}\qquad 
		\inferrule*[right=$\mathcal{U}$-F]{\Gamma \vdash \phi:\propo \\ \Gamma\vdash \psi:\propo}{\Gamma  \vdash \phi \mathcal{U}\psi:\propo}
		\end{equation*}
		For any context $\Gamma$ we define $\formu{\Sigma}(\Gamma)$ to be the set of formulae $\phi$ such that $\Gamma \vdash \phi:\propo$. 
	\end{definition}

	Then, we can introduce the rules for the logical judgments of the form $\Gamma \mid \Phi \vdash \phi$, where $\Phi$ is a finite set of propositions well-formed in $\Gamma$.
	\begin{definition}\label{def:slcsrules}
		We define four rules for the well-formed formulae previously defined:
		\begin{itemize}
			\item $\mathcal{C}$'s rules:
			\begin{equation*}	
			\inferrule*[right=Cl-$1$]{\Gamma \mid \Phi\vdash \psi}{\Gamma \mid \Phi \vdash \mathcal{C}(\psi) }\quad
			\inferrule*[right=Cl-$2$]{\Gamma \mid \Phi, \psi \vdash \phi }{\Gamma  \mid \Phi, \mathcal{C}(\psi ) \vdash \mathcal{C}(\phi)}
			\end{equation*}
			\iffalse 
			\item $\mathcal{C}$'s rules:
			\begin{equation*}	
			\begin{gathered}
			\inferrule*[right=Cl-$1$]{\hspace{1pt}\hspace{1pt}}{\Gamma \mid \Phi, \mathcal{C}(\bot)\vdash \bot}\\
			\inferrule*[right=Cl-$3$]{\hspace{1pt}\hspace{1pt} }{\Gamma  \mid \Phi, \mathcal{C}(\phi\vee \psi ) \vdash \mathcal{C}(\phi)\vee \mathcal{C}(\psi)}
			\end{gathered}\quad
			\begin{gathered}
			\inferrule*[right=Cl-$2$]{\hspace{1pt}\hspace{1pt} }{\Gamma  \mid \Phi, \phi \vdash\mathcal{C}(\phi) }	\\
			\inferrule*[right=Cl-$4$]{\hspace{1pt}\hspace{1pt}}{\Gamma \mid  \Phi, \mathcal{C}(\phi)\vee \mathcal{C}(\psi) \vdash\mathcal{C}(\phi\vee \psi ) } 
			\end{gathered}
			\end{equation*}
			\fi 
			\item $\mathcal{U}$'s rules \vspace{-2ex}
			\begin{gather*}
			\inferrule*[right=$\mathcal{U}$-I]{\Gamma \mid \Phi, \varphi \vdash \phi\\\Gamma \mid \Phi, \mathcal{C}(\varphi), \neg \phi \vdash \psi}{\Gamma \mid \Phi, \varphi \vdash \phi\mathcal{U}\psi}
			\\	
			\inferrule*[right=$\mathcal{U}$-E]
			{\text{for all } \phi \text{ such that }  \Gamma \vdash \varphi:\propo:\quad 
				\Gamma \mid \Phi, \varphi \Rightarrow \phi, (\mathcal{C}(\varphi)\wedge \neg\varphi) \Rightarrow \psi,\varphi \vdash \theta}{\Gamma \mid \Phi, \phi\mathcal{U}\psi  \vdash \theta }
			\end{gather*}
		\end{itemize}
		
		The \emph{Propositional Logic for Closure Operators on $\Sigma$} (PLCO) is given by the usual propositional rules (i.e., without the quantifiers) for the typed (intuitionistic) sequent calculus (see e.g.~\cite{jacobs1999categorical}), extended with the four rules above.
		
	The \emph{First Order Logic for Closure Operators on $\Sigma$} (FOLCO) is given by the four rules above added to the usual rules for first order logic. Similarily with equality.	
		
		Derivability of sequents is defined in the usual way \cite{pitts1995categorical}.
	\end{definition}
	\begin{remark}
		PLCO corresponds to the Spatial Logic for Closure Spaces considered in 	\cite{ciancia2014specifying}.
	\end{remark}
	\begin{remark}
		Notice that $\mathcal{U}$-E is an \emph{infinitary} rule saying that a formula $\theta$ can be derived from $\phi\mathcal{U}\psi$ if it can be derived from \emph{all} the formulae $\varphi$ satisfying precise conditions.
		Thus, this rule shows the second-order nature of the $\mathcal{U}$ operator.
	\end{remark}

	\subsection{Categorical semantics of closure logics}
	In this section we provide a sound and complete categorical semantics of the logics for the closure operators defined above.
	
	\begin{definition}
		Two formulae $\phi,\psi \in \formu{\Sigma}(\Gamma)$ are \emph{provably equivalent} if $\Gamma \mid \psi \vdash \phi$ and $\Gamma \mid \phi \vdash \psi$.
		We will denote the quotient of  $\formu{\Sigma}(\Gamma)$ by this relation with $\lind{\Sigma}(\Gamma)$, $[\phi]$ will denote the class of $\phi$ in it.
	\end{definition}
	
	\begin{theoremEnd}{prop}
		For any signature $\Sigma$ the following are true:
		\begin{enumerate}
			\item $\lind{\Sigma}(\Gamma)$ equipped with the order $[\phi]\leq [\psi]$ if and only if $\Gamma \mid \phi\vdash \psi$ is derivable is:
			\begin{itemize}
				\item a meet semilattice in the case we are considering regular logic;
				\item  a Heyting algebra if we are considering propositional or first order logic;
			\end{itemize} 
			\item $[\phi \mathcal{U}\psi]$ is
			the supremum of the set 
			\begin{equation*}
			\mathsf{u}_{\Gamma}(\phi, \psi):=\{[\varphi]\in \lind{\Sigma}(\Gamma) \ \mathrm{such \ that}\  \Gamma \mid \varphi \vdash \phi, \Gamma \mid \mathcal{C}(\varphi), \neg \varphi \vdash \psi\}
			\end{equation*}
			\item 
			there exists a (elementary) closure or existential doctrine or a (elementary) hyperdoctrine  $(\lind{\Sigma}, \mathfrak{c}_\Sigma)$ on $\class{\Sigma}$ sending $\Gamma$ to $\lind{\Sigma}(\Gamma)$.
		\end{enumerate}
	\end{theoremEnd}
	\begin{proof}
		\begin{enumerate}
			\item 	The logical connectives induce a Heyting algebra or a meet semilattice structure on $\lind{\Sigma}(\Gamma)$ which has precisely $\leq$ as associated order.
			\item From $\mathcal{U}$-I follows that $[\phi \mathcal{U}\psi]$ is an upper bound for $	\mathsf{u}_{\Gamma}$ while $\mathcal{U}$-E	implies that $[\phi \mathcal{U}\psi]$ is the least of them.
			\item For any morphism $(T_1,...,T_n):\Gamma \rightarrow \Delta$ substitution of terms gives us a morphism of Heyting algebras/meet semilattices $\lind{\Sigma}(\Delta)\rightarrow \lind{\Sigma}(\Gamma)$; quantifiers gives us the existential doctrine/hyperdoctrine structure (cfr. \cite{pitts1995categorical} for the details). In any case have to define a preclosure operator $\mathfrak{c}_{\Sigma, \Gamma}$ on each $\lind{\Sigma}(\Gamma)$ but this is easily done defining
			\begin{align*}
			\mathfrak{c}_{\Sigma, \Gamma}: \lind{\Sigma}(\Gamma)&\rightarrow \lind{\Sigma}(\Gamma)\\
			[\phi] & \mapsto [\mathcal{C}(\phi)]
			\end{align*}
			The $\mathcal{C}$'s rules assure us that $\mathfrak{c}_\Sigma$ is well defined, inflationary and monotone, while an easy induction shows that 
			\begin{align*}
			\lind{\Sigma}_{(T_1,...,T_n)}([\mathcal{C}(\phi)])&=	\mathfrak{c}_{\Sigma, \Gamma}(\lind{\Sigma}_{(T_1,...,T_n)}(\phi))
			\end{align*}
			for any $(T_1,...,T_n):\Gamma\rightarrow \Delta$.
			We can add fibered equalities, given $\Gamma:=[x_i:\sigma_i]$ putting:
			\begin{equation*}
			\delta_{\Gamma \times \Gamma}:=\bigwedge_{i=1}^n[ x_i=_{\sigma_i}y_i]
			\end{equation*}
			where $\{y_i\}_{i=1}^n$ is a set of fresh variables such that $y_i:\sigma_i$ for any $i$.
			\qedhere
		\end{enumerate} 
	\end{proof}
	
	Let us prove the soundness and completeness of the categorical semantics wrt.~the various logical fragments.
	\begin{definition}\label{def:slcsmodel}
		Let $(\mathscr{P}, \mathfrak{c}):\catname{C}^{op}\rightarrow \catname{InfSL}$ be an (elementary) closure doctrine (existential doctrine/hyperdoctrine) then a morphism of $\catname{cPD}$ ($\catname{cED}$, $\catname{cEED}$, $\catname{cEHD}$, $\catname{cHD}$) $(\mathscr{M}, \mu): (\lind{\Sigma}, \mathcal{C})\rightarrow (\mathscr{P}\mathfrak{c})$ is a \emph{model of the propositional (first-order) logic (with equality) of closure operators in $(\mathscr{P}, \mathfrak{c})$} if it is open. 
		
		\iffalse and
		\begin{equation*}
		\mu_\Gamma([\phi \mathcal{U}\psi])=\bigvee_{[\varphi]\in [\mathsf{u}_\Gamma](\phi, \psi)}\mu_\Gamma([\varphi])
		\end{equation*}
		for any two formulae $\phi$ and $\psi$ such that $[\phi]$ and $[\psi]\in \lind{\Sigma}(\Gamma)$.\fi  
		
		A sequent $\Gamma \mid \Phi\vdash \psi$ is \emph{satisfied by $(\mathscr{M}, \mu)$} if 
		\begin{equation*}
			\bigwedge_{\phi \in \Phi}\mu_{\Gamma}(\phi) \leq \mu_{\Gamma}(\psi)
		\end{equation*}
	\end{definition}
	
	\begin{theorem}\label{th:completeness}
		A sequent $\Gamma \mid \Phi\vdash \psi$ is satisfied by the \emph{generic model} $(1_{\class{\Sigma}}, 1_{\lind{\Sigma}})$ 
		if and only if it is derivable.
	\end{theorem}
	\begin{proof}
		By definition, $\Gamma \mid \Phi\vdash \psi$ is satisfied if and only if 
		\begin{equation*}
		\bigwedge_{\phi \in \Phi}[\phi]\leq [\psi]
		\end{equation*}
		in $\lind{\Sigma}(\Gamma)$ but this is equivalent to the derivability of
		\begin{equation*}
		\Gamma \mid \bigwedge_{\phi \in \Phi}\phi \vdash \psi
		\end{equation*}
		whose derivability is equivalent (applying the conjunction rules a finite number of times) to
		$		\Gamma \mid \Phi \vdash \psi $,
		and we are done.\qedhere
	\end{proof}
	
	\begin{corollary}
		The above defined categorical semantics for PLCO or FOLCO (with or without equality) is sound and complete.
	\end{corollary}
	\begin{proof}
		The only thing left to show is soundness for an arbitrary $(\mathscr{P}, \mathfrak{c})$ but this follows at once since each component $\mu_\Gamma$ of $\mu$ is monotone. %
		\qedhere
	\end{proof}
	
	\subsection{About the semantics of $\mathcal{U}$}\label{sur}
	As we have remarked before, the rule $\mathcal{U}$-E for the operator $\mathcal{U}$ is infinitary.
	Although in general this is needed, in this section we will define a class of hyperdoctrines in which the semantics of $\mathcal{U}$ can be given as a supremum of approximants. %
	\begin{definition}\label{def:extbound}
		Let $(\mathscr{P}, \mathfrak{c}):\catname{C}^{op}\rightarrow \catname{InfSL}$ be a closure doctrine that factors through the category of Heyting algebras. 
		For any object $C$ define the \emph{external boundary}:
		\begin{align*}
		\partial^+_C:\mathscr{P}(C)&\rightarrow \mathscr{P}(C)\\
		\alpha &\mapsto \mathfrak{c}_C(\alpha)\wedge \neg \alpha 
		\end{align*}
		For $\phi$ and $\psi\in \mathscr{P}(C)$, we define $\phi\mathfrak{U}_C\psi\in \mathscr{P}(C)$ as the supremum, if it exists, of the set
		\begin{equation*}
		\mathfrak{u}_C(\phi, \psi):=	\{\varphi\in \mathscr{P}(C) \mid \varphi\leq \phi \text{ and } \partial^+_C(\varphi) \leq\psi \}
		\end{equation*} 
	\end{definition}
	
	\begin{remark}
		If $\mathscr{P}$ is $\lind{\Sigma}$ then
		$[\phi]\mathfrak{U}_\Gamma[\psi]=[\phi\mathcal{U}\psi]$
		for any $[\phi]$ and $[\psi]\in \lind{\Sigma}(\Gamma)$.
	\end{remark}
	\begin{remark}\label{inclu}
		If $(\mathscr{M}, \mu)$ is a model then 
		$\mu_\Gamma(\mathsf{u}_{\Gamma}(\phi, \psi))\subset  \mathfrak{u}_{\mathscr{M}(\Gamma)}(\mu_\Gamma([\phi]), \mu_\Gamma([\psi]) )$
		for any $\Gamma$.
	\end{remark}
	
	\begin{example}\label{exa}
		Let $(X,\mathfrak{c})$ be a pretopological space and $S$, $T\in \mathcal{P}^p(X,\mathfrak{c})$, then
		\begin{equation*}
		S\mathfrak{U}_{(X, \mathfrak{c})}T = \bigcup \{W\subset S\mid \partial^+_{(X,\mathfrak{c})}(W) \subset T\}
		\end{equation*}
		i.e. $x\in S\mathfrak{U}_{(X, \mathfrak{c})}T$ if and only if there exists $W\subset S$ such that $X\in W$ and $\partial^+_{(X,\mathfrak{c})}(W) \subset T$.
	\end{example}
	\begin{example}
		Let us consider the closure operator $\mathfrak{c}_\epsilon$ on $\catname{Set}(-, [0,1])$ (see \cref{sec:reprex}).
		For any $f:X\rightarrow [0,1]$, it is $(\neg f)(x)= 1$ if and only if $f(x)=0$. 
		So,
		\[(\mathfrak{c}_{X, \epsilon}(f)\wedge \neg f)(x)=\begin{cases}
		\epsilon  & f(x)=0 \\
		0 & f(x)\neq 0
		\end{cases},
		\] 
		hence, given $g, h:X\rightarrow [0,1]$, $f\in \mathsf{u}_\Gamma(g,h)$ if and only if $f\leq g$ and $h(x)\geq \epsilon$ for any $x\in f^{-1}(0)$.
	\end{example}

	\begin{remark}
		If $(\mathscr{M}, \mu)$ is a model then for any $[\varphi]\in \lind{\Sigma}(\Gamma)$ such that $\varphi \in \mathsf{u}_\Gamma(\phi, \psi)$ we have 
		$\mathscr{\mu}_\Gamma([\varphi])\leq \mu_{\Gamma}([\phi \mathcal{U}\psi])$.
	\end{remark}
	
	\begin{definition}
		Let $(\mathscr{P}, \mathfrak{c})$ be as in \cref{def:extbound}. 
		A model $(\mathscr{M}, \mu):\lind{\Sigma}\rightarrow (\mathscr{P},\mathfrak{c})$ is said \emph{continuous} if the equality
		\begin{equation*}
		\mu_{\Gamma}([\phi \mathcal{U}\psi])=\mu_{\Gamma}([\phi])\mathfrak{U}_{\mathscr{M}(\Gamma)}\mu_{\Gamma}([\psi])
		\end{equation*}
		holds for any context $\Gamma$ and $[\phi], [\psi]\in \lind{\Sigma}(\Gamma)$.
	\end{definition}

	\begin{proposition}
		Let $\Sigma$ be a signature and $(\mathscr{P}, \mathfrak{c})$ a complete (elementary, existential, or hyper)doctrine, i.e. $\mathscr{P}(C)$ is complete for any object $C$ of $\catname{C}$; then, for any product preserving functor: $\mathscr{M}:\class{\Sigma} \rightarrow \catname{C}$ and functions 
		\begin{equation*}
		\mu^*_{\Gamma}:\Pi(\sigma_1,...,\sigma_n)\rightarrow \mathscr{P}(\mathscr{M}(\Gamma))
		\end{equation*}
		for all $\Gamma=[x_i:\sigma_i]_{i=1}^n$, there exists a unique continuous model $(\mathscr{M}, \mu)$ in $(\mathscr{P}, \mathfrak{c})$ such that
		\begin{equation*}
		\mu_\Gamma([P(x_1,...,x_n)])=\mu^*_\Gamma(P)
		\end{equation*}
	\end{proposition}
	\begin{proof}
		This follows immediately by induction. \qedhere
	\end{proof}

	\begin{example}
		Let $\mathcal{X}=\{(X_i, \mathfrak{c}_{i})\}_{i\in I}$ be a small family of pretopological spaces and let us define $\Sigma$ as follows:
		\begin{gather*}
		\abs{\Sigma}:=\mathcal{X}\qquad  \Gamma(((X_{i_1}, c_{i_1}),...,(X_{i_n},c_{i_n})), (X_j, c_{j})):=\catname{PrTop}(\prod_{k=1}^{n}(X_{i_k},c_{i_k}), (X_j, c_j)) \\ \Pi((X_{i_1}, c_{i_1}),...,(X_{i_n},c_{i_n})):=\mathcal{P}(\prod_{k=1}^{n}X_{i_k}) 
		\end{gather*}
		We can take as $\mathscr{M}$ the unique product preserving functor  
		$\class{\Sigma}\rightarrow \catname{PrTop}$ such that
		\begin{gather*}
		\functor[l]{(X_i,c_i)}{f}{(X_i,c_i)}
		\functormapsto
		\functor[r]{(X_i,c_i)}{f}{(X_i,c_i)}
		\end{gather*}
		i.e., $\mathscr{M}$ sends contexts to products and lists of terms to the corresponding product arrow. We can define $\mu^*$ sending each predicate $P:(X_{i_1},\mathfrak{c}_{i_1}),...,(X_{i_n},\mathfrak{c}_{i_n})$ to corresponding subset of $\prod_{k=1}^{n}$ $(X_{i_k},\mathfrak{c}_{i_k})$.
		\cref{exa} guarantees that this semantics is the same as the one developed in \cite{ciancia2014specifying}.
	\end{example}
	
	\begin{proposition}
		For any signature $\Sigma$ a sequent is derivable if and only if it is satisfied by any continuous model.
	\end{proposition}
	\begin{proof}
		This follows immediately by the fact that the generic model is continuous. \qedhere
	\end{proof}
	
	\subsection{Higher order SLCS}
	We can provide an analogous of the above results for higher order logic.
	\begin{definition}
		A \emph{higher order signature} $\Sigma$ is a triple $(\abs{\Sigma}, \propo, \Gamma )$ where
		\begin{itemize}
			\item $\abs{\Sigma}$ is a set, called the set of \emph{basic types};
			\item $\propo \in \abs{\Sigma}$;
			\item $\Gamma:\abs{\Sigma}^\star \times \abs{\Sigma}\rightarrow \catname{Sets}$ is a functor, we will call \emph{function symbol} an element $f$ of $\Gamma((\sigma_1,...,\sigma_n), \sigma_{n+1})$ and we will write $f:\sigma,...\sigma_n\rightarrow, \sigma_{n+1}$.
		\end{itemize}
		A morphism of signature $\phi:\Sigma_1\rightarrow \Sigma_2$ is a couple $(\phi_1,\phi_2)$ such that
		\begin{itemize}
			\item $\phi_1$ is a function $\abs{\Sigma_1}\rightarrow \abs{\Sigma_2}$ such taht $\phi_1(\propo_1)=\propo_2$;
			\item $\phi_2$ is a natural transformation $\Gamma_1\rightarrow \Gamma_2\circ (\phi_1^{\star}\times \phi_1)$.
		\end{itemize}
		For any $\sigma\in \abs{\Sigma}$ we fix an countably infinite set $X_\sigma$ of \emph{variables}
	\end{definition}
	
	\begin{definition}
		We define the \emph{set of types} $\typ{\Sigma}$ as the smallest set such that:
		\begin{itemize}
			\item $\abs{\Sigma}\cup \{1\}\subset \typ{\Sigma}$;
			\item if $\sigma$ and $\tau\in \typ{\Sigma}$ then $\sigma \times \tau$ and $\sigma\rightarrow \tau \in \typ{\Sigma}$.
		\end{itemize}
		The introduction and elimination rules for product and exponent type and the rules for contexts are the usual ones (\cite{jacobs1999categorical}), like before we must add to them the following two:
		\begin{equation*}
		\inferrule*[right=$\mathcal{C}$-F]{\Gamma \vdash \phi:\propo}{\Gamma \vdash \mathcal{C}(\phi):\propo}\qquad 
		\inferrule*[right=$\mathcal{U}$-F]{\Gamma \vdash \phi:\propo \\ \Gamma\vdash \psi:\propo}{\Gamma  \vdash \phi \mathcal{U}\psi:\propo}
		\end{equation*}
		
	\end{definition}
	\begin{remark}
		Notice that now $\phi:\propo$ doesn't mean that $\phi$ is well-formed but that it is a term of type $\propo$.
	\end{remark}
	\begin{definition} Terms are constructed and typed in the usual way (\cite{jacobs1999categorical}, given two types $\sigma$ and $\tau$, the \emph{set of terms from $\sigma$ to $\tau$} $\term{\sigma}{\tau}$ is the quotient of the set of terms $T$ such that  $x:\sigma \vdash T:\tau$ modulo provable equality, where equality is subjected to the following rules (\cite{jacobs1999categorical} ): 
		\begin{gather*}
		\begin{gathered}
		\inferrule*[right=$\beta$-Con]{\Gamma, v:\sigma \vdash M:\tau \\ \Gamma:N:\sigma}{\Gamma \vdash \lambda v:\sigma.M(N)=_\tau M[N|v] }\\
		\inferrule*{\Gamma \vdash M:1}{\Gamma \vdash M=_1(\hspace{1pt}) }\\
		\inferrule*{\Gamma \vdash M:\sigma \\ \Gamma \vdash N:\tau}{\Gamma \vdash \pi_\sigma(M,N)=_\sigma M }
		\end{gathered}\qquad
		\begin{gathered}
		\inferrule*[right=$\eta$-Con]{\Gamma \vdash M:\sigma \rightarrow \tau}{\Gamma \vdash \lambda v:\sigma.M(v)=_{\sigma \rightarrow \tau}M}	\\
		\inferrule*{\Gamma \vdash P:\sigma\times \tau}{\Gamma \vdash (\pi_\sigma(P), \pi_\tau(P))=_{\sigma \times \tau} P }
		\\
		\inferrule*{\Gamma \vdash M:\sigma \\ \Gamma \vdash N:\tau}{\Gamma \vdash \pi_\tau(M,N)=_\tau N }
		\end{gathered}
		\end{gather*}	
	\end{definition}
	\begin{definition}
		Given an higher order signature $\Sigma$, its \emph{classifying category} is the category $\hoclass{\Sigma}$ in which
		\begin{itemize}
			\item the set of objects is $\typ{\Sigma}$;
			\item a morphism $\sigma \rightarrow \tau$ is just an element of $\term{\sigma}{\tau}$;	
			\item composition is given by substitution.
		\end{itemize}
	\end{definition}
	\begin{theoremEnd}{prop}\label{cc}
		For any higher order signature $\Sigma$, $\hoclass{\Sigma}$ is a cartesian closed category.
	\end{theoremEnd}
	\begin{proof}\begin{proofEnd}
			Cfr. \cite{jacobs1999categorical}, proposition $2.3.2$. The existence of finite product is clear, let's show the closedness. Projecting and evaluating we have a sequent
			$w:(\sigma\rightarrow\tau)\times \sigma \vdash \pi_{\sigma\rightarrow \tau}(w)(\pi_\sigma(w)):\tau$, let $ev_\sigma$ be the class of $\pi_{\sigma\rightarrow \tau}(x)(\pi_\sigma(x))$ in $\term{(\sigma\rightarrow \tau) \times \sigma}{\tau}$ and let's show its universality. 
			Let $\rho$ be another type and $[M]\in \term{\rho \times \sigma}{\tau}$ then we must have a sequent $z:\rho \times \sigma \vdash M:\tau$ and, weakening, we can derive
			\begin{gather*}
			\inferrule*{\inferrule*{x:\rho, y:\sigma, z:\rho \times \sigma \vdash M:\tau\\x:\rho, y:\sigma \vdash (x,y):\rho \times \sigma }
				{x:\rho , y:\sigma \vdash M[(x,y)|z]: \tau}}
			{x:\rho \vdash \lambda y:\sigma.M[(x,y)|z]:\sigma \rightarrow \tau}
			\end{gather*}
			Let's define 
			\begin{equation*}
			\Lambda([M]):=[\lambda y:\sigma.M[(x,y)|z]\in \term{\rho}{\sigma\rightarrow \tau}
			\end{equation*}
			Computing we have
			\begin{align*}
			ev_{\sigma}\circ (	\Lambda([M])\times 1_{\sigma})&=[\pi_{\sigma\rightarrow \tau}(w)(\pi_\sigma(w))[( (\lambda y:\sigma.M[(x,y)|z])[p_\rho(z)|x],p_\sigma(z))|w]]\\&=[(\lambda y:\sigma.M[(p_\rho(z),y)|z])(p_\sigma(z))]\\&=[M[(p_\rho(z),p_\sigma(z))|z]]\\&=[M]
			\end{align*}
			For uniqueness it's enough to notice that for any $[N]\in \term{\rho}{\sigma \rightarrow \tau}$:
			\begin{align*}
			\Lambda(ev_{\sigma}\circ ([N]\times 1_\sigma))&=[\lambda y:\sigma.\pi_{\sigma\rightarrow \tau}(w)(\pi_\sigma(w))[(N,y)|w] ]\\&=[\lambda y:\sigma.N(y)]=[N]
			\qedhere\end{align*}
		\end{proofEnd}\qedhere
	\end{proof}
	\begin{definition}
		The \emph{higher order spatial logic for closure spaces on a signature $\Sigma$} is the typed intuitionistic sequent calculus built on $\Sigma$ with the same rules of \cref{def:slcsrules}.
	\end{definition}
	
	\begin{theoremEnd}{lem}\label{rap}
		For any higher order signature $\Sigma$ the following are true:
		\begin{enumerate}
			\item $(\term{-}{\propo}, \mathscr{c}_{\Sigma})$ is a closure hyperdoctrine;
			\item for any $\phi$, $\psi\in \term{\sigma}{\propo}$ 
			\begin{equation*}
			[\phi] \leq [\psi] \iff x:\sigma \mid \phi \vdash \psi
			\end{equation*}
		\end{enumerate}
	\end{theoremEnd}
	\begin{proof}
		\begin{proofEnd}
			\begin{enumerate}
				\item The type formation rules and for $\propo$ tell us that  $\term{-}{\propo}$ is a functor $\hoclass{\Sigma}^{op}\rightarrow \catname{HA}$ in which the closure operator is given by
				\begin{align*}
				\mathfrak{c}_\Sigma: \term{\sigma}{\propo} &\rightarrow \term{\sigma}{\propo}\\
				[\phi] &\mapsto [\mathcal{C}(\phi)]
				\end{align*} 
				The fibered equality is given by the sequent
				\begin{gather*}
				(M,N):\sigma \times \sigma \vdash M=_\sigma N :\propo 
				\end{gather*}
				Now notice that the rules of derivation tell us that
				\begin{align*}
				\term{\sigma \times \tau}{\propo} &\rightarrow \term{\sigma}{\propo}\\
				\phi &\mapsto \exists_{x:\tau}(\phi) 
				\\[1ex]
				\term{\sigma \times \tau}{\propo} &\rightarrow \term{\sigma}{\propo}\\
				\phi &\mapsto \forall_{x:\tau}(\phi) 
				\end{align*}
				are the left and right adjoint to 
				\begin{align*}
				\term{\pi_\sigma}{\propo}:\term{\sigma}{\propo} &\rightarrow \term{\sigma \times \tau}{\propo}\\
				\phi &\mapsto \phi \circ \pi_{\sigma}
				\end{align*}
				and that Frobenius reciprocity and the Beck-Chevalley condition hold, and so we can conclude.
				\item This follows immediately by the rules for conjunction. \qedhere
			\end{enumerate} 
		\end{proofEnd}\qedhere  
	\end{proof}

	\begin{definition}
		Let $(\mathscr{P}, \mathfrak{c}):\catname{C}^{op}\rightarrow \catname{HA}$ be a closure hyperdoctrine with $\mathscr{P}$ isomorphic to $\catname{C}(-,C)$ with $\catname{C}$ cartesian closed. Then, a morphism of $\catname{cEHD}$
		\[
		(\mathscr{M}, \mu): (\term{-}{\propo}, \mathfrak{c}_\Sigma) \rightarrow (\mathscr{P},\mathfrak{c})
		\]
		is a \emph{model of SLCS in $\mathscr{P}$} if:
		\begin{itemize}
			\item $\mathscr{M}(\propo)\simeq C$;
			\item $\mathscr{M}$ preserves exponentials;
			\item $(\mathscr{M}, \mu)$ is open.
		\end{itemize}
		A sequent $x:\sigma  \mid \Phi\vdash \psi$ is \emph{satisfied by $(\mathscr{M}, \mu)$} if 
		\begin{equation*}
		\bigwedge_{\phi \in \Phi}\mu_{\sigma}([\phi]) \leq \mu_{\sigma}([\psi]).
		\end{equation*}

	\end{definition}
	
	\begin{theorem}
		A sequent $\Gamma \mid \Phi\vdash \psi$ is satisfied by the \emph{generic model} given by the identities $(1_{\hoclass{\Sigma}}, 1_{\term{-}{\propo}})$ if and only if it is derivable.
	\end{theorem}
	\begin{proof}
		By definition, $x:\sigma \mid \Phi\vdash \psi$ is satisfied if and only if 
		\begin{equation*}
		\bigwedge_{\phi \in \Phi}[\phi] \leq [\psi]
		\end{equation*}
		in $\term{\sigma}{\propo}$ but this is equivalent to the derivability of
		\begin{equation*}
		x:\sigma \mid \bigwedge_{\phi \in \Phi}\phi \vdash \psi
		\end{equation*}
		but the derivability of it is equivalent (applying the rules of conjunction a finite number of times) to
		\begin{equation*}
		x:\sigma \mid \Phi \vdash \psi 
		\end{equation*} 
		and we are done.
	\end{proof}
	
	\begin{corollary}
		The categorical semantics of higher order $SLCS$ is sound and complete.
	\end{corollary}
	\begin{proof}
		The only thing left to show is soundness for an arbitrary $(\mathscr{P}, \mathfrak{c})$ in $\catname{cEHD}$ with $\mathscr{P}$ representable functor on a cartesian closed category, but this follows at once since each component of $\mu$ is monotone.
	\end{proof}
	\begin{remark}
		We can repeat verbatim the considerations of subsection \ref{sur} in order to get a completeness result for continuous models.
	\end{remark}
	
	\ifreport
	
	\section{Paths in closure doctrines}\label{sec:paths}
	Often, in spatial logics we are interested also on \emph{reachability} of some property.
	Differently from closure and the ``until'' operator, reachability is not a structural property of the logical domain; rather, it depends on the kind of paths we choose to explore the space.
	In this section we formalise this idea, and show how to interpret also the $\mathcal{S}$ operator from SLCS.
	
	\subsection{The reachability closure operator}
	\begin{definition}
		Let $\mathscr{P}:\catname{C}^{op}\rightarrow \catname{HA}$ be an hyperdoctrine, an \emph{internal preorder} in $\mathscr{P}$ is a pair $(I, \rho)$ where $I$ is an object of $\catname{C}$ and $\rho\in \mathscr{P}(I\times I)$ which satisfy:
		\begin{itemize}
			\item  reflexivivty: $\delta_I \leq \rho$;
			\item transitivity: $\mathscr{P}_{(\pi_1,\pi_2)}(\rho)\wedge\mathscr{P}_{(\pi_2,\pi_3)}(\rho)\leq \mathscr{P}_{(\pi_1,\pi_3)}(\rho)$ 
		\end{itemize} 
		$(I,\rho)$ is called an \emph{internal order} if in addition $\rho $ is \emph{antisymmetric}, i.e.
		$\rho\wedge\mathscr{P}_{(\pi_2,\pi_1)}(\rho) \leq \delta_I$.
		Moreover $(I, \rho)$ is \emph{total} if $\rho \vee \mathscr{P}_{(\pi_2,\pi_1)} (\rho) = \top$.
		
		A \emph{internal monotone arrow} $f:(I,\rho)\rightarrow (J, \sigma)$ is an arrow of $\catname{C}$ such that $\rho \leq \mathscr{P}_{f\times f}(\sigma)$.
	\end{definition}
	
	\begin{definition}
		Let $(\mathscr{P}, \mathfrak{c}):\catname{C}^{op}\rightarrow \catname{HA}$ be an elementary existential closure doctrine, we say that $\phi\in \mathscr{P}(C)$ is \emph{connected} if $\varphi \vee \psi = \phi$ and  $\mathfrak{c}(\varphi)\wedge \psi = \bot $ imply $\varphi=\bot$.
		
		An object $C$ is \emph{$\mathscr{P}$-connected} if $\top\in \mathscr{P}(C)$ is connected.
	\end{definition}
	
	\begin{definition}
		Given a preorder $(I, \rho)$ in an elementary existential doctrine $\mathscr{P}$ and $\alpha\in \mathscr{P}(I)$ we define the \emph{downward} and \emph{upward closure of $\alpha$} as
		\begin{equation*}
		\downarrow \alpha:= \exists_{\pi_1}(\mathscr{P}_{\pi_2}(\alpha)\wedge \rho )\qquad \uparrow \alpha:= \exists_{\pi_2}(\mathscr{P}_{\pi_1}(\alpha)\wedge \rho )
		\end{equation*} 
		We define the \emph{reachability operator} $\reach{}$ as the family of functions, indexed over the objects: 
		\begin{align*}
		\reach{C}:\mathscr{P}(C)&\rightarrow \mathscr{P}(C)\\
		\varphi &\mapsto \varphi \vee \bigvee_{p\in \catname{C}(I,C)}\exists_{p}(\uparrow\mathscr{P}_{p}(\varphi)).
		\end{align*}
	\end{definition}
	
	\begin{theoremEnd}{prop}
		Given $\mathscr{P}:\catname{C}^{op}\rightarrow \catname{InfSL}$ in $\catname{EED}$ and $(I,\rho)$ an internal preorder in it, $\reach{}=\{\reach{C}\}_{C\in \catname{Ob}(\catname{C})}$  is a grounded closure operator on $\mathscr{P}$. If, moreover, $\mathscr{P}$ is an hyperdoctrine, $\reach{}$ is fully additive.
	\end{theoremEnd}
\begin{proof}
	\begin{proofEnd}
		Monotonicity and inflationarity comes at once, take an arrow $f:C\rightarrow D$, for any $\alpha \in \mathscr{P}(D)$ we have:
		\begin{align*}
		\exists_f(\reach{C}(\mathscr{P}_f(\varphi)))&=\exists_f(\mathscr{P}_f(\varphi))\vee \bigvee_{p\in \catname{C}(I,C)}\exists_f(\exists_p(\uparrow \mathscr{P}_p(\varphi)))
		\\&=\exists_f(\mathscr{P}_f(\varphi))\vee \bigvee_{p\in \catname{C}(I,C)}\exists_f(\exists_{p}(\exists_{\pi_2}(\rho \wedge \mathscr{P}_{\pi_1}(\mathscr{P}_{p}(\mathscr{P}_{f}(\varphi))))))\\&\leq
		\varphi \vee \bigvee_{p\in \catname{C}(I,C)} \exists_{f\circ p}(\exists_{p}(\exists_{\pi_2}(\rho \wedge \mathscr{P}_{\pi_1}(\mathscr{P}_{f\circ p}(\varphi))))\\&\leq \varphi \vee \bigvee_{q\in \catname{C}(I,D)} \exists_{q}(\exists_{\pi_2}(\rho \wedge \mathscr{P}_{\pi_1}(\mathscr{P}_{q}(\varphi))))\\&=\varphi \vee \bigvee_{q\in \catname{C}(I,D)}\exists_q(\uparrow \mathscr{P}_q(\varphi)) \\&=\reach{D}(\varphi)
		\end{align*}
		Groundedness is immediate; suppose now that $\mathscr{P}$ is an hyperdoctrine, then $\mathscr{P}_f$ commutes with suprema for any arrow $f$ and, since $\mathscr{P}(C)$ is an Heyting algebra, infima distribute over them, so:
		\begin{align*}
		\reach{C}(\bigvee_{k\in K}\varphi_k)&= (\bigvee_{k\in K}\varphi_K )\vee \bigvee_{p\in \catname{C}(I,C)} \exists_{p}(\exists_{\pi_1}(\rho \wedge \mathscr{P}_{\pi_2}(\mathscr{P}_p(\bigvee_{k\in K}\varphi_k))))\\&=(\bigvee_{k\in K}\varphi_k)\vee \bigvee_{p\in \catname{C}(I,C)}\bigvee_{k\in K}\exists_{p}(\exists_{\pi_1}(\rho \wedge \mathscr{P}_{\pi_2}(\varphi_k)))\\&=\bigvee_{k\in K}(\varphi_k\vee \bigvee_{p\in \catname{C}(I,C)}\exists_{p}(\exists_{\pi_1}(\rho \wedge \mathscr{P}_{\pi_2}(\varphi_k))))\\&= \bigvee_{k\in K} \reach{C}(\varphi_k)
		\qedhere 
		\end{align*}
	\end{proofEnd}\qedhere 
\end{proof}
	\begin{example}
		In $\catname{Set}$, for any non empty $I$ we have that, for any set $X$:
		\begin{gather*}
		\reach{X}(S)=\begin{cases}
		X & S\neq \emptyset\\
		\emptyset & S=\emptyset
		\end{cases}
		\end{gather*}  
	\end{example}
	
	\begin{example}
		Take the elementary hyperdoctrine $\mathscr{P}^p$ on $\catname{PrTop}$ (\cref{def:topclosures}) and fix an $n\in \mathbb{N}$, as an internal order we can take $n=\{0,1,...,n-1\}$ with the closure operator
		\begin{align*}
		\mathfrak{n}:\mathcal{P}(n)&\rightarrow \mathcal{P}(n)\\
		S & \mapsto \{i\in n \mid n=s+1 \text{ for some } s \in S\} 
		\end{align*} and the usual ordering $\leq$ as $\rho$. An arrow $p:(I, \mathfrak{n})\rightarrow (X, \mathfrak{c})$ is just a function such that
		\begin{equation*}
		\mathfrak{n}(p^{-1}(S))\leq p^{-1}(\mathfrak{c}(S))
		\end{equation*}
		that is, $p(i+1)\in \mathfrak{c}(\{p(i)\})$.
		So, for instance
		\begin{equation*}
		\reach{(\mathbb{N}, \mathfrak{n})}(S)=\{k \in \mathbb{N}\mid k=s+n \text{ for some } s\in S \}
		\end{equation*}
		where $\mathsf{n}:\mathcal{P}(\mathbb{N})\rightarrow \mathcal{P}(\mathbb{N})$ is defined as for $n$.
	\end{example}
	
	\subsection{Surroundedness}
	In this section we will introduce a \emph{surrounded} operator (similar to the ``until'' operator of temporal logic) in order to generalize the analogous operator introduced in \cite{ciancia2016spatial}. 
	\begin{definition}
		Let $(I,\rho)$ be an internal order in an elementary existential closure doctrine $(\mathscr{P}, \mathfrak{c})$, let $\phi$ and $\psi\in \mathscr{P}(C)$.
		We say that an arrow $p:I\rightarrow C$ is an \emph{escape route from $\phi$ avoiding $\psi$} if
		\begin{enumerate}
			\item at some point in $p$, $\phi$ holds: $\exists_{!_I}(\mathscr{P}_p(\phi))=\top$;
			\item from the points where $\phi$ holds we can reach a point where $\neg\phi$ holds: $\mathscr{P}_p(\phi)\leq {\downarrow}\mathscr{P}_p(\neg \phi)$;
			\item there is no point reachable from $\phi$ and which reaches $\neg\phi$ along the route, where $\psi$ holds:
			$\uparrow\mathscr{P}_p(\phi)\wedge \downarrow \mathscr{P}_p(\neg \phi)\wedge \mathscr{P}_p(\psi)=\bot$.
		\end{enumerate}
		We will denote with $\er{C}{\phi}{\psi}$ the set of such arrows. We also define
		\begin{equation*}
		\escape{C}{\phi}{\psi} :=\bigvee_{p\in \er{C}{\phi}{\psi}}  \exists_{p}(\top)
		\qquad\qquad
		\sur{C}{\phi}{\psi}:=\bigwedge_{p\in \er{C}{\phi}{\psi}}\phi \wedge \neg(\exists_p(\top))
		\end{equation*}
	\end{definition}
	Intuitively, $\escape{C}{\phi}{\psi}$ (read ``$\phi$ escapes $\psi$'') holds where $\phi$ holds and it is possible to escape avoiding $\psi$;
	conversely, $\sur{C}{\phi}{\psi}$  (read ``$\phi$ is surrounded by $\psi$'') holds where $\phi$ holds and it is not possible to escape from it without avoiding $\psi$. 
	Notice that these notions depend on the specific choice of the internal order $(I,\rho)$, hence we can deal with different reachability, with different shapes of escape routes, by choosing the adequate internal order.

	\begin{example}[cfr.~\cite{ciancia2016spatial}]
		Let us  consider the closure hyperdoctrine on pretopological spaces $(\mathscr{P}^p, c)$ as in \cref{def:topclosures}. In this case an internal order is just an ordered set $(I,\leq)$ equipped with a closure operator. Given $S$ and $T$ subsets of a chosen $(X, \mathfrak{c})$, then
		\begin{itemize}
			\item $p\in \er{(X, \mathfrak{c})}{S}{T}$ if and only if
			\begin{enumerate}
				\item $p^{-1}(S)\neq \emptyset$;
				\item for any $t$ such that $p(t)\in S$ there exists an $s\geq t$ with $p(s)\notin S$; 
				\item  $p(t)\notin T$ for any $t\in I$ for which there exist $s$ and $v\in I$ such that $p(s)\in A$, $p(v)\in T $ and $s \leq t \leq v$.
			\end{enumerate}
			\item   $x \in 	\escape{(X, \mathfrak{c})}{S}{T} $ if and only if there exists a continuous $p:I\rightarrow X, t, s \in I$ such that $t\leq s$, $p(t)=x$, $p(s)\notin S$ and for any pair $(u,v)\in \leq$ with $p(u)\in S$ and $p(v)\in T$ there are no $w$ between $u$ and $v$ such that $p(w)\in T$.
			\item $x\in \sur{(X, \mathfrak{c})}{S}{T}$ if and only if $x\in S$  and for any continuous $p:I\rightarrow X$ such that $p(t)=x$ for some $t\in I$, $p\notin \er{(X, \mathfrak{c})}{S}{T}$.
		\end{itemize}   
		Therefore, this situation corresponds  to the surround operator defined in   \cite{ciancia2016spatial}.
	\end{example}
	
	\begin{theoremEnd}{thm1}
		Let $(\mathscr{P}, \mathfrak{c})$ be a boolean elementary closure hyperdoctrine, $(I, \rho)$ a preorder in it with $I$ $\mathscr{P}$-connected and such that, for all $\gamma\in \mathscr{P}(I)$,
		$\mathfrak{c}_I(\gamma)\wedge\neg \gamma \leq \uparrow \gamma$.
		Then, for any $\phi$ and $\psi \in \mathscr{P}(C)$:
		\begin{enumerate}
			\item if $\alpha\in \mathfrak{u}_{C}(\phi, \psi)$  and $p\in \er{C}{\phi}{\psi}$ then $\mathscr{P}_p(\alpha)=\bot$;
			\item $\phi \mathfrak{U}_C \psi \leq \sur{C}{\phi}{\psi}$.
		\end{enumerate}
	\end{theoremEnd}
	\begin{proof}
	\begin{proofEnd}
		\begin{enumerate}
			\item By continuity we have
			\begin{align*}
			\mathfrak{c}_I(\mathscr{P}_p(\alpha))\wedge \mathscr{P}_p(\neg \alpha)&\leq \mathscr{P}_p(\mathfrak{c}_C(\alpha)) \wedge \mathscr{P}_p(\neg \alpha)\\&\leq \mathscr{P}_p(\mathfrak{c}_C(\alpha)\wedge \neg \alpha)\\&\leq		
			\mathscr{P}_p(\psi)
			\end{align*}
			By hypothesis,
			\begin{align*}
			\mathfrak{c}_I(\mathscr{P}_p(\alpha))\wedge \mathscr{P}_p(\neg \alpha)&\leq \uparrow\mathscr{P}_p(\alpha)\\&\leq \mathscr{P}_p(\phi)
			\end{align*} 
			and 
			\begin{align*}
			\mathfrak{c}_I(\mathscr{P}_p(\alpha))\wedge \mathscr{P}_p(\neg \alpha)&=\mathfrak{c}_I(\mathscr{P}_p(\alpha))\wedge \mathscr{P}_p(\neg \alpha)\wedge \top\\&=(\mathfrak{c}_I(\mathscr{P}_p(\alpha))\wedge \mathscr{P}_p(\neg \alpha)\wedge \mathscr{P}_{p}(\phi) )\vee(\mathfrak{c}_I(\mathscr{P}_p(\alpha))\wedge \mathscr{P}_p(\neg \alpha)\wedge \mathscr{P}_p(\neg \phi))\\&\leq\downarrow\mathscr{P}_p(\neg \phi)
			\end{align*}
			hence, since
			$p\in \er{C}{\phi}{\psi}$:
			\begin{align*}
			\mathfrak{c}_I(\mathscr{P}_p(\alpha))\wedge \mathscr{P}_p(\neg \alpha)&\leq \uparrow\mathscr{P}_p(\phi)\wedge \downarrow \mathscr{P}_p(\neg \phi)\wedge \mathscr{P}_p(\psi) = \bot
			\end{align*}	
			and we conclude by connectedness.	
			\item By the previous point $\mathscr{P}_p(\alpha)=\bot$ for any $p\in \er{C}{\phi}{\psi}$ so $\mathscr{P}_p(\neg \alpha)=\top$ that implies $\alpha \leq \neg \exists_p(\top)$ from which the thesis follows.
			\qedhere\end{enumerate}
	\end{proofEnd}
	\qedhere 
\end{proof}

	\section{Logics for closure hyperdoctrines with paths}\label{sec:slcswp}
	In this section we extend the logics for closure hyperdoctrines we have introduced in \cref{sec:slcs}, with formulae constructor for reasoning about sourroundedness and reachability.
	
	\subsection{Syntax and derivation rules}
	\begin{definition}
		A \emph{signature with paths} is a triple $\Sigma=(\Sigma, \iota, R)$ where
		\begin{itemize}
			\item $\Sigma$ is a signature as per \cref{def:signature};
			\item $\iota \in \abs{\Sigma}$ is called the \emph{interval type};
			\item $R: \iota, \iota$ is called the \emph{preorder of $\iota$}
		\end{itemize}
		A morphism  $\phi:(\Sigma_1, \iota_1, R_1)\rightarrow (\Sigma_2, \iota_2, R_2)$ is a morphism of signature  $(\phi_1,\phi_2,\phi_3)$ such that $\phi_1(\iota_1)=\iota_2$ and ${\phi_3}_{\iota, \iota}(R_1)=R_2$.
	\end{definition}
	\begin{remark}
		Signatures with paths and their morphisms with componentwise composition form a category $\catname{SignPath}$.
	\end{remark}
	
	\begin{definition}
		We add the following rule of well formation to the logic for the closure operators (\cref{def:slcssynt}):
		\begin{equation*}
		\inferrule*[right=$\mathcal{S}$-F]{\Gamma \vdash \phi:\propo \\
			\Gamma\vdash \psi:\propo}{\Gamma  \vdash \phi \mathcal{S}\psi:\propo}
		\end{equation*}
		
	\end{definition}
	
	\begin{definition}
		Given a signature $(\Sigma, \iota, R)$, its \emph{classifying category} is the category $\class{\Sigma, \iota, R}$ is just $\class{\Sigma}$.  
	\end{definition}

	\begin{definition}
		We define the following rules for the well-formed formulae previously defined:
		\begin{itemize}
			\item $R$'s rules:
			\begin{gather*}	
			\inferrule*[right=$R$-Refl]{\Gamma, x:\iota, y:\iota \mid \Phi\vdash x=_\iota y}{\Gamma, x:\iota, y:\iota \mid \Phi\vdash R(x,y)} \\
			\inferrule*[right=$R$-Trans]{\Gamma, x:\iota, y:\iota \mid \Phi\vdash R(x,y)\\\Gamma, y:\iota, z:\iota \mid \Phi\vdash R(y,z)}{\Gamma, x:\iota, z:\iota \mid \Phi\vdash R(x,z)}
			\end{gather*}
			
			\item $\mathcal{S}$'s rules: 
			\begin{gather*}
			\inferrule*[right=$\mathcal{S}$-I]{\text{for all } p:\iota \rightarrow \sigma,\quad  \Gamma, x:\sigma  \mid \Phi, \varphi \vdash \es{\Gamma,x:\sigma}{\phi}{\psi}\wedge \phi \wedge \neg \exists t:\iota.(x=_\sigma p(t))
			}{\Gamma, x:\sigma \mid \Phi, \varphi \vdash 	\phi\mathcal{S}\psi}  	
			\\
			\inferrule*[right=$\mathcal{S}$-E]{\Gamma, x:\sigma \mid \Phi \vdash \es{\Gamma,x:\sigma}{\phi}{\psi}}{\Gamma, x:\sigma  \mid  \Phi,\phi\mathcal{S}\psi  \vdash \phi \wedge \neg \exists t:\iota.(x=_\sigma p(t))}  
			\end{gather*} 
			where
			\begin{align*}
			\es{\Gamma,x:\sigma}{\phi}{\psi} :=\ & (\exists t:\iota.\phi[p(t)/x])\ \wedge \\
			& (\phi[p(t)/x] \Rightarrow \exists s:\iota.(R(t, s)\wedge \neg \phi[p(s)/x]))\ \wedge \\
			& \neg (\exists s:\iota.(R(s,t)\wedge \phi[p(s)/x])\ \wedge \\
			& \quad \exists v:\iota.(R(t,v)\wedge \neg \phi[p(v)/x]) \wedge \psi[p(t)/x])
			\end{align*}

		\end{itemize}
		
		The \emph{Propositional Logic for Closure Operators with Paths on $\Sigma$} (PLCOwP) is given by PLCO (\cref{def:slcsrules}) extended with the rules above.
		Similarly for the \emph{First Order Logic for Closure Operators with Paths on $\Sigma$} (FOLCOwP). %
		Derivability of sequents is defined in the usual way (\cite{pitts1995categorical}).
	\end{definition}

	\subsection{Categorical semantics of closure logics with paths}
	\begin{definition} 
		Given an elementary closure hyperdoctrine $(\mathscr{P}, \mathfrak{c}):\catname{C}^{op}\rightarrow \catname{HA}$ and an internal preorder $(I, \rho)$, we will call the pair $((\mathscr{P}, \mathfrak{c}), (I,\rho))$ an \emph{elementary path hyperdoctrine}.
		An arrow of path hyperdoctrines $((\mathscr{P}, \mathfrak{c}), (I,\rho))\rightarrow ((\mathscr{S}, \mathfrak{d}), (J,\sigma))$ is a morphism $(\mathscr{F}, \eta )\in \catname{cEHD}((\mathscr{P}, \mathfrak{c}), (\mathscr{S}, \mathfrak{d}) )$ such that there exists an isomorphism $h:\mathscr{F}(I)\rightarrow J$ for which
		\begin{equation*}
		\eta_{I\times I}(\rho)=\mathscr{S}_{(h\circ \mathscr{F}(\pi_1), h\circ \mathscr{F}(\pi_2) )}(\sigma)
		\end{equation*}  
		We say that $(\mathscr{F}, \eta)$ is \emph{open} if it is as arrow $(\mathscr{P}, \mathfrak{c})\rightarrow (\mathscr{S}, \mathfrak{d})$
		
		Clearly this defines a $2$-subcategory $\catname{pEHD}$ of $\catname{cEHD}$.
	\end{definition}
	\begin{proposition}For any signature $(\Sigma, \iota, R)$, $(\lind{\Sigma}, (\iota, R))$ is a path hyperdoctrine.
	\end{proposition}
	\begin{proof} 
		We have only to show that $(\iota, R)$ is an internal preorder but this follows at once from the two $R$'s rules.
	\end{proof}

	\begin{definition}
		Let $((\mathscr{P},\mathfrak{c})(I,\rho))$ be a path hyperdoctrine.
		Then, a \emph{model of closure logic with paths} in it is just an open  morphism
		\begin{equation*}
		(\mathscr{M}, \mu):((\lind{\Sigma}, \mathfrak{c}_\Sigma), (\iota, R)\rightarrow ((\mathscr{P},\mathfrak{c}),(I,\rho))
		\end{equation*}
		Satisfability of sequents is defined as in the case of closure logics (\cref{def:slcsmodel}).
	\end{definition}
	\begin{remark}
		As for $\mathcal{U}$ we have not put any requirement on the interpretation of $\mathcal{S}$, but, in $(\lind{\Sigma}, \mathfrak{c}_\Sigma)$, for $\Gamma \vdash \phi:\propo$ and $\Gamma \vdash \psi:\propo$ we have
		\begin{equation*}
		[\phi\mathcal{S}\psi]=[\phi]\mathfrak{S}_\Gamma[\psi]
		\end{equation*}
		so we can again ask for \emph{continuous models}, i.e. models that preserves this equality.
	\end{remark}
	\iffalse 
	\begin{proposition}
		For any pair of formulae $\phi$ and $\psi$ such that $\Gamma \vdash \phi:\propo$ and $\Gamma \vdash \psi: \propo$ and any model $(\mathscr{M}, \mu)$ the equalities
		\begin{equation*}
		\mu_{\Gamma}([\phi\mathcal{S}\psi])=\Su{(I,\rho)}{\mathscr{M}(\Gamma)}(\mu_{\Gamma}([\phi]),\mu_{\Gamma}([\psi]))\qquad 	\mu_{\Gamma}([\phi\mathcal{P}\psi])=\Pro{(I,\rho)}{\mathscr{M}(\Gamma)}(\mu_{\Gamma}([\phi]),\mu_{\Gamma}([\psi]))
		\end{equation*}
		hold.
	\end{proposition}
	\begin{proof}
		It is enough to notice that
		\begin{equation*}
		\begin{split}
		\mu_{\Gamma}([\phi\mathcal{S}\psi])&=	\mu_{\Gamma}([\surr(\phi, \psi)])\\&=\Su{(I,\rho)}{\mathscr{M}(\Gamma)}(\mu_{\Gamma}([\phi]),\mu_{\Gamma}([\psi]))
		\end{split}
		\begin{split}
		\mu_{\Gamma}([\phi\mathcal{P}\psi])&=	\mu_{\Gamma}([\propag(\phi, \psi)])\\&=\Pro{(I,\rho)}{\mathscr{M}(\Gamma)}(\mu_{\Gamma}([\phi]),\mu_{\Gamma}([\psi]))
		\end{split}
		\end{equation*}
	\end{proof}
	\fi
	\begin{theorem}
		A sequent $\Gamma \mid \Phi\vdash \psi$ is satisfied by the \emph{generic model} $(1_{\class{\Sigma}}, 1_{\lind{\Sigma}})$ if and only if it is derivable.
	\end{theorem}
	\begin{proof}
		The proof is the same as for \cref{th:completeness}.
	\end{proof}
	
	\begin{corollary}
		The above defined categorical semantics for PLCOwP/RLCOwP/FOLCOwP (with or without equality) is sound and complete.
	\end{corollary}
	
	\fi
	\section{Conclusions and future work}
	\label{sec:concl}
	In this paper we have introduced \emph{closure (hyper)doctrines} as a theoretical framework for studying the logical aspects of closure spaces.
	First we have shown the generality of this notion by means of a wide range of examples arising naturally from topological spaces, fuzzy sets, algebraic structures, coalgebras, and covering at once also known cases such as Kripke frames and probabilistic frames.
	Then, we have applied this framework to provide the first axiomatisation and sound and complete categorical semantics for various fragments of a logic for closure doctrines. In particular, the propositional fragment corresponds to the Spatial Logic for Closure Spaces \cite{ciancia2014specifying}, a modal logic for the specification and verification on spatial properties over preclosure spaces. But the flexibility of our approach allows us to readily obtain closure logics for a wide range of cases (including all the examples presented above).
	
	\ifreport
	Finally, we have extended closure hyperdoctrines with a notion of \emph{paths}. This allows us to provide sound and complete logical derivation rules also for the ``surroundedness'', thus covering all the logical constructs of SLCS (including ``reachability'').
	\fi

	Albeit already quite general, the theory presented in this paper paves the way for several extensions. 
	\ifreport
	\else
	Due to lack of space, we have not been able to present the constructions for modeling logical operators concerning \emph{surroundedness}. To this end, we need to endow doctrines with an object representing the ``type of paths''; for more details we refer to the extended version of this work \cite{cm:closurehyperdoctrines-extended}.
	
	\fi
	We can enrich the logic with other spatial modalities, e.g., the spatial counterparts of the various temporal modalities of CTL* \cite{emerson1986sometimes}.
	It could be interesting to investigate a spatial logic with fixed points \emph{a la} $\mu$-calculus; to interpret such a logic, we could consider closure hyperdoctrines over Löb algebras. 
	Moreover, it would be interesting to develop some ``generic'' model checking algorithm for spatial logic. The abstraction provided by the categorical approach can guide the generalization of existing model checking algorithms, such as \cite{ciancia2014specifying}, and suggest new proof methodologies and minimisation techniques.
	
	On a different direction, we are interested in the type theory induced by closure hyperdoctrines.
	A Curry-Howard isomorphism would yield a functional programming language with constructors for spatial aspects, which would be very useful in \emph{collective spatial programming}, e.g. for collective adaptive systems.

	\appendix
	\section{Omitted proofs}\label{sec:proofs}
	\printProofs
\end{document}